\documentclass{lmcs}
\pdfoutput=1

\usepackage{lastpage}
\lmcsdoi{17}{2}{19}
\lmcsheading{}{\pageref{LastPage}}{}{}%
{Feb.~20,~2020}{May~26,~2021}{}

\usepackage[utf8]{inputenc}


\usepackage{amssymb}
\usepackage{amsmath}
\usepackage{graphicx}

\usepackage{algorithmic}
\usepackage[linesnumbered,ruled,noend]{algorithm2e}
\DontPrintSemicolon

\newcommand{\definitie}{{\overset{\mathrm{def}}{=}}}

\begin{document}

\title[Interface Modeling for Quality and Resource Management]{Interface Modeling for Quality and Resource Management\rsuper*}

\titlecomment{{\lsuper*}The work described in this paper was supported from the ECSEL Joint Undertaking (JU) under grant agreement No 783162 (FitOpTiVis). The JU receives support from the European Union's Horizon 2020 research and innovation programme and Netherlands, Czech Republic, Finland, Spain, and Italy.}

\author[M.~Hendriks]{Martijn Hendriks{\rsuper{{a,b}}}}
\author[M.~Geilen]{Marc Geilen\rsuper{a}}
\author[K.~Goossens]{Kees Goossens\rsuper{a}}
\author[R.~de Jong]{Rob de Jong\rsuper{c}}
\author[T.~Basten]{Twan Basten{\rsuper{{a,b}}}}
\address{{\lsuper a}Eindhoven University of Technology, Eindhoven, The Netherlands}
\address{{\lsuper b}ESI (TNO), Eindhoven, The Netherlands}
\address{{\lsuper c}Philips Medical Systems International BV, Best, The Netherlands}

\begin{abstract}
We develop an interface-modeling framework for quality and resource management that captures configurable working points of hardware and software components in terms of functionality, resource usage and provision, and quality indicators such as performance and energy consumption.
We base these aspects on partially-ordered sets  to capture quality levels, budget sizes, and functional compatibility. This makes the  framework widely applicable and domain independent  (although we aim for embedded and cyber-physical systems).
The framework paves the way for dynamic (re-)configuration and multi-objective optimization of component-based systems for quality- and resource-management purposes.
\end{abstract}

\maketitle

\section{Introduction}
Component frameworks enable the construction of systems from a set of loosely coupled, independent components through composition mechanisms.
A component framework emphasizes abstraction, separation of concerns, and reuse.
It has inherently loose coupling between components, and explicit and visible dependencies.
This eases, for instance, dynamic reconfiguration and thereby adaptability.
Component-based models of computation can be classified into component models and interface models~\cite{AH01b}.
A component model specifies how a component behaves in an arbitrary environment.
An interface model specifies how a component can be used, e.g., for composition, and restricts the environment: it makes input assumptions and provides output guarantees for the case when the input assumptions are met by the environment.
We present an interface model for (dynamic) quality and resource management (QRM) of component-based systems.
Each component in our framework has one or more \emph{configurations} that may be parameterized. A configuration has an \emph{input}, an \emph{output}, a \emph{required budget}, a \emph{provided budget}, and a \emph{quality}.
Parameters capture the configurable working points of the component.
External actors, the user or a quality and resource manager for instance, can use the parameters to partially control configuration of components for, e.g., dynamic reconfiguration of the system.
The input and output model the functionality of a particular configuration of the component, to the extent relevant for QRM. For instance, a component configuration can have a video stream as input, and for each frame a set of image features as output. The types of extracted features may be adjustable via the parameters.
The required and provided budgets model the services that a component configuration provides or requires.
A component configuration may, for instance, require, or conversely, offer, a volatile storage capacity of 64~MB. 
Finally, the quality of a component configuration refers to the properties that we want to optimize, for instance, latency, throughput, energy consumption, cost. Qualities, budgets, and outputs typically depend on inputs and parameter settings.
The five parts and the parameters form a \emph{quality- and resource-management interface} that can be used for dynamic reconfiguration of the system to ensure functionality and a certain quality of service.
We use concepts from the Pareto Algebra framework~\cite{pareto-algebra} to mathematically define our components and their composition.

\paragraph{Related work.}
Design based on components is a well-known approach to tackle the increasing complexity of today's systems.
The general idea is to use pre-defined components with well-specified behaviors and interfaces as interoperable building blocks. These components can be composed to create larger components with more complex functionality.
Well-known component models in the software domain are the Component Object Model (COM)~\cite{com},
the Common Object Request Broker Architecture (CORBA)~\cite{corba},
JavaBeans~\cite{javabeans},
and Open Services Gateway Initiative (OSGi)~\cite{osgi}. These are all generally applicable and do not depend on the application domain.
An example of a more specialized software component model is the KOALA framework~\cite{koala}.
It is aimed at the software development of a product family of high-end television sets, with an emphasis on dealing with variability.
Platforms that aim at predictability and composability are, for instance, the CompSOC platform~\cite{compsoc} and PRET machines~\cite{pret}. CompSOC aims to reduce system complexity through composability of applications and through (timing) predictability. PRET machines also focus on providing solid timing predictability for processors, thereby making the timing of the composition of programs predictable.
Examples of component frameworks with tool support for cyber-physical systems are the Behavior Interaction Priorities (BIP) framework~\cite {BIB} and the {Ptolemy} framework~\cite{Ptolemy}.
Other conceptual component models are, for instance, the I/O automata~\cite{LT87} and timed I/O automata~\cite{KL+03} frameworks, which provide formal notions of components, composition, and abstraction.

Interface models for components have been introduced by~\cite{AH01b}.
An interface model of a component specifies what the component expects of the environment and what it can provide. This supports compositional refinement and therefore component-based design.
 Using this principle, interface automata capture temporal assumptions about how a component is used and how it uses other components~\cite{AH01}. The authors of, e.g.,~\cite{WT05, HM06,comma,dfrefine} use interface-based design concepts for real-time aspects of systems.
Contract-based design, e.g.,~\cite{B08,cbd17}, uses similar ideas. It specifies contracts of components as a combination of assumptions that the component makes about the environment and the guarantees that it then can give.

All work mentioned above has the scope of software development or system design.
In contrast, we position our work in the scope of {system operation} with QRM as focus.
In~\cite{AP+16}, a middleware solution is presented that builds on a general-purpose component model. The middleware supports dynamic reconfiguration of components based on their quality-of-service demands. The quality-of-service parameters are fixed to period, deadline and priority, and each deployed component has a known worst-case execution time.
The view presented in~\cite{AP+16} of components as sets of instances, each with its own quality-of-service parameters, is similar to our view.
A runtime resource manager based on OpenCL is introduced in~\cite{LM+13}.
An extended compiler generates OpenCL kernels of the same code for various computational resources (CPU, GPU, FPGA, etc.) including resource-usage information. The runtime resource manager uses this to dynamically adapt the system if needed.
The authors of ~\cite{GBE10} informally introduce a component framework for quality-of-service management. They distinguish functional requirements, e.g., the required input format of a video-processing component, from resource requirements, e.g., the amount of memory needed. This is similar to our view in which we also make this distinction explicit.
Finally, \cite{compsoc2} presents a resource-management framework to create and manage multi-processor partitions at runtime, which provides a high level of (timing) predictability. The resource model is a simple instance of our model.

\paragraph{Contribution.}
Our framework uses partially-ordered sets as the interface model of components and defines composition in terms of operations from the Pareto Algebra~\cite{pareto-algebra}.
We show that the framework supports a notion of refinement called \emph{dominance}.
This natively supports domain-independent quality and resource management by a formal semantics and compositional multi-objective optimization.
Our work is on a formal, conceptual level and this distinguishes it from existing component-based approaches within the quality- and resource-management scope, e.g.,~\cite{AP+16,LM+13,GBE10}, as these are concretely tailored to some domain or lack a formal semantics. The generality of our approach enables application to many domains, but requires further specialization before it can be applied.

\paragraph{Outline.}
We use a motivating example of a video-processing system from the healthcare domain throughout this paper. This example is introduced in Sec.~\ref{sec:ex}.
Sec.~\ref{sec:pa} summarizes the main results from~\cite{pareto-algebra} on Pareto Algebra that form the foundation of our interface model for QRM.
In Sec.~\ref{sec:components} we introduce our interface model, and in Sec.~\ref{sec:composition} we define  composition of interfaces based on Pareto-algebraic operations.
This is then detailed in Sec.~\ref{sec:recipe2} and applied to the example in 
Sec.~\ref{sec:qrm}.
In Sec.~\ref{sec:qrml}, we introduce the QRML component language and toolset for the specification and analysis of quality- and resource-management problems and explain the semantics in terms of the QRM interface model developed in this paper.
Finally, Sec.~\ref{sec:conclusions} concludes.


\section{Motivating Example -- Description}\label{sec:ex}
\begin{exa}\label{ex:1}
We consider a video-processing system from the healthcare domain.
The system scales video streams in an operating theatre, for display on one or more displays. The system receives video streams with a certain resolution and rate. Examples of common resolutions are qHD (960x540 pixels), HD (1280x720 pixels), HD+ (1600x900 pixels) and FHD (1920x1080 pixels).
We assume that a pixel is encoded with four bytes.
The platform of the system consists of fiber links, to transport streams from their sources to processing resources, and hardware {scalers} that perform the processing.
Each fiber has a capacity of 10~Gb/s. A single uncompressed FHD stream at 60~Hz takes approximately 4~Gb/s, so a fiber can transport at most two such streams.
A scaler can downscale a stream to a given output resolution. Each scaler can process up to 300~Mpixels/s, and it can handle up to four streams.
Depending on the resolution of a stream, a scaler requires a certain amount of memory to buffer a single horizontal line of that stream.
Each scaler has access to only a limited amount of memory that thus must be shared by the streams that are mapped to that scaler.
This memory is split into 32 segments of 128 pixels each.
A scaler can reduce the frame rate of the outgoing stream to reduce its computational requirements. For instance, if a scaler does not have enough computational resources to process an incoming FHD stream at 60~Hz, then it can reduce the outgoing rate from 60 to 30~Hz by dropping half of the incoming video frames. This reduces the video quality at the benefit of being able to process more streams. The end user can also prioritize streams.
The task of the quality and resource manager of this system is to map streams to scalers while respecting the resource constraints, and at the same time optimizing the frame rate of the outgoing streams given the end-user's priority.
Dynamic reconfiguration takes place when streams are added or removed, or when the end user changes output resolutions or priorities.

\begin{figure}\label{fig:ex1}
\centering
\includegraphics[width=\linewidth]{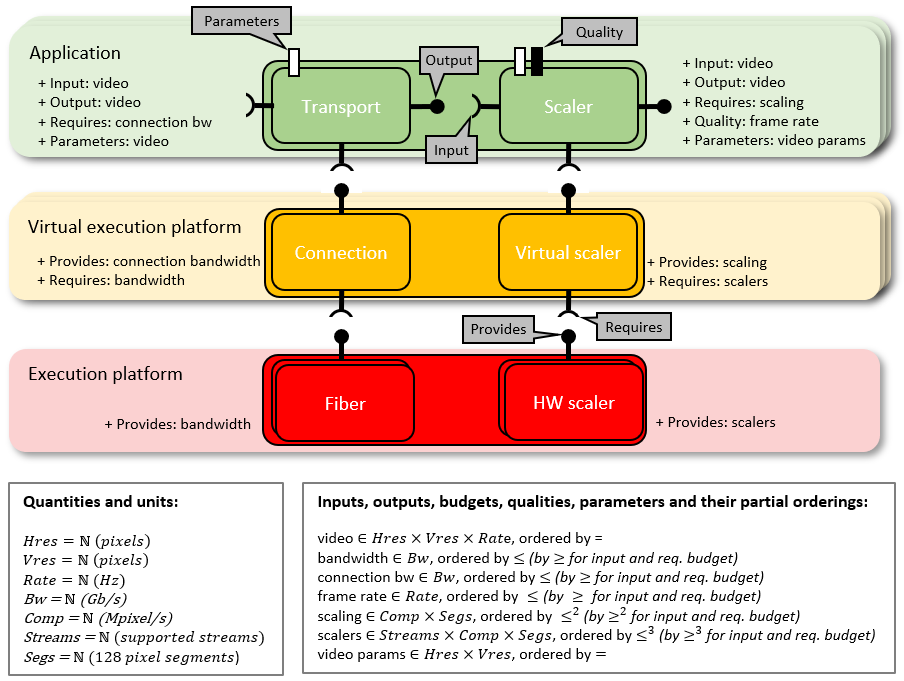}
\caption{An example video-processing system with a single stream.}\label{fig:video}
\end{figure}

Figure \ref{fig:video} shows an overview of the components in the system.
Each stream in the system is modeled as an \emph{application} that consists of a Transport component and a Scaler component.
The Transport component has video both as input and output and requires connection bandwidth that is determined by the rate and resolution of the video that is transported. This is a parameter of the Transport component that is set at runtime to match the transport configuration to the actual incoming stream.
Video is modeled as a horizontal resolution in pixels, a vertical resolution in pixels, and a rate in~Hz, each being a natural number. The connection bandwidth is the number of bits per second, again a natural number.
The Scaler component also has video as input and output, and requires a scaling budget, which consists of computational capacity (Mpixels per second) and memory requirements (a number of 128 pixel segments for line buffering).
The parameters of the Scaler specify working points for the resolution of the resulting video stream, and these can be set from the outside by, e.g., the end user  or some application using the video system.
The computational budget needed by the Scaler is determined by the input video resolution and the output video resolution and rate.
The quality aspect of the Scaler is the frame rate of the outgoing video stream.

To have a feasible composition, the output of the Transport needs to \emph{match} the input of the Scaler. We use partial-order relations to model this i/o matching, and also for the budget matching, i.e., matching a required budget with a provided budget.
For instance, the $=$ relation could be used for \emph{video}, as the Scaler can only guarantee its correct functioning if it exactly knows the input resolutions and rate.
For connection bandwidth, for instance, we can use the $\le$ relation, as the transportation of a certain video can be guaranteed also when more bandwidth is given to the Transport component than strictly needed.

A Scaler component typically has multiple \emph{configurations} that the end user or the quality and resource manager can partially select through the parameters of the component.
A Transport component also has multiple configurations, where the specific configuration is determined by the environment since the type of the incoming stream is uncontrollable.
For instance, suppose that we need to scale a 60~Hz FHD stream.
(We abbreviate video-related types; for instance, we write FHD@60 for $(1920,1080,60)$.)
Then we have a Transport component with a  configuration that requires a connection-bandwidth budget of $4$~Gb/s, with a FHD@60 input and output video.
The Scaler component has parameterized configurations. One configuration associated with output-resolution parameter HD has FHD@60 video input and HD@60 video output, requires a scaling budget of $(171,15)$, and has frame rate $60$ as quality.
Another configuration, also associated with output-resolution parameter HD, has FHD@60 video input and HD@30 video output, requires a scaling budget of $(145,15)$, and has frame rate $30$ as quality.
A third configuration with parameter HD has FHD@60 video input and HD@20 video output with scaling budget $(136,15)$, and quality $20$. The latter two configurations trade quality for a smaller required budget.

The resource needs of an application, i.e., a Transport and Scaler pair in this example, are specified via a required private \emph{virtual execution platform} (VEP) that abstracts from the underlying physical execution platform (for instance, real-time operating system and hardware). A VEP for our stream application consists of a Connection component that provides a connection bandwidth, and a Virtual Scaler component that provides scaling.
The underlying hardware forms the \emph{execution platform}, and consists of Fiber components and HW scaler components.
A single Fiber provides bandwidth in~Gb/s, which is modeled as a natural number.
Each Scaler provides a budget consisting of three parts: (i) the number of streams that can still be processed by this Scaler (limited to four), (ii) the free amount of memory (at most 32 segments), and (iii) the remaining computational power (at most 300~Mpixels/s), captured by the tuple $(4, 32, 300)$.

Again, suppose that a 60~Hz FHD stream needs to be scaled. The Transport component requires a connection bandwidth budget of $4$~Gb/s, so a Connection component is used that provides this budget.
To realize this, it requires a bandwidth budget of $4$~Gb/s from some Fiber component. The remaining $6$~Gb/s of that Fiber is available for other VEPs.
In the same way, the required and provided budget of the Virtual scaler component can be calculated from the Scaler component (for a specific configuration of the Scaler component). A Virtual scaler component can be composed with any HW scaler component with matching budget.

In general, composition may lead to configurations that are not interesting because other configurations perform better in all aspects.
Consider, for instance, two 60~Hz FHD streams that need to be scaled to HD. There are many possibilities to compose these with the hardware platform. A naive way is to compose both streams with a single HW scaler. This can be done if the output rate is set to 30~Hz.
Another way is to compose each stream with its own Fiber and HW scaler. In that case, the output rates can be 60~Hz.
The first option is \emph{dominated} by the second option if the used and remaining platform resources are not important.
Such a dominated combination can be ignored for further analysis and composition. That is, only the \emph{Pareto-optimal} component combinations need to be considered.

The quality- and resource-management problem for this system is the following. Given a number of incoming video streams (applications) and given the required output resolution of each video stream (parameters), find mappings of these streams on a platform consisting of a number of Fibers and a number of HW scalers such that the output rates (qualities) are optimized according to some user-defined cost function (e.g., maximize the minimum output rate  among all streams).
In the following sections, we introduce the mathematics that we need to solve this problem. The application of the theory is demonstrated later in the paper, with a domain-specific solution for the quality- and \allowbreak{}resource-management problem of the video-processing system explained above.
\end{exa}


\section{Preliminaries -- Pareto Algebra}\label{sec:pa}
This section briefly summarizes the results from Pareto Algebra~\cite{pareto-algebra} that we need for our interface-modeling framework. Prop.~\ref{prop:derivation-constraint} is the only new part that we have added.

A partially-ordered set (poset) is a set $Q$ with a partial-order relation $\preceq_Q$.
In the remainder, we take shortcuts for readability and conciseness in the sense that we often write about a poset $Q$ without mentioning the accompanying partial-order relation $\preceq_Q$ explicitly.
Furthermore, we sometimes omit the subscript $Q$ from $\preceq_Q$ when the ordering relation is clear from the context.
The converse relation of $\preceq$ is denoted by $\preceq^{-1}$ and the strict (i.e., non-reflexive) version of the relation is denoted by $\prec$.

\begin{defi}[Configuration Space]\label{def:configspace}
A configuration space $\mathcal{S}$ is defined as the Cartesian product  $Q_1 \times Q_2 \times \ldots \times Q_n$ of a finite number of posets. A configuration is an element of a configuration space.
\end{defi}

It is useful to have a notion whether a configuration is better than another configuration.
This kind of information can help to choose between configurations in quality- and resource-management questions.
The notion of \emph{dominance} on a configuration space serves this purpose.

\begin{defi}[Dominance]\label{def:dominance}
Let $\mathcal{S}$ be a configuration space, and let $c = (q_1, \ldots, q_n), c' = (q_1', \ldots, q_n') \in \mathcal{S}$. We say that $c$ is dominated by $c'$, or equivalently, that $c'$ dominates $c$, denoted by $c \preceq c'$, if and only if $q_i \preceq_{Q_i} q_i'$ for all $1 \le i \le n$.
\end{defi}

Note that we reverse the interpretation of $\preceq$ with respect to~\cite{pareto-algebra}.
Observe that a configuration space with the dominance relation constructed from posets is itself again a poset. Hence, configuration spaces can serve as posets in other configuration spaces.

The notion of Pareto minimality expresses that a set of configurations only contains configurations that are relevant.
This can be used to define the concepts of set dominance and Pareto equivalence of configuration sets.

\begin{defi}[Pareto minimality]\label{def:min}
Let $\mathcal{S}$ be a configuration space and let $C \subseteq \mathcal{S}$ be a set of configurations. We say that $C$ is {Pareto minimal} if and only if $c \not\prec c'$ for all $c, c' \in C$.
\end{defi}

\begin{defi}[Set dominance]\label{def:cdom}
Let $\mathcal{S}$ be a configuration space and let $C, C' \subseteq \mathcal{S}$ be sets of configurations.
We say that $C'$  {dominates} $C$, denoted by $C \preceq C'$, if and only if for every $c \in C$ there is some $c' \in C'$ such that $c \preceq c'$.
\end{defi}

\begin{defi}[Pareto equivalence]\label{def:equiv}
Let $\mathcal{S}$ be a configuration space and let $C, C' \subseteq \mathcal{S}$.
We say that $C$  {is equivalent to} $C'$, denoted by $C \equiv C'$ if and only if $C \preceq C' \land C' \preceq C$.
\end{defi}

Pareto equivalence expresses the notion that neither of the sets have a configuration that cannot be at least matched by a configuration of the other set.
The following theorem states that every well-ordered\footnote{A poset $(Q, \sqsubseteq)$ is well ordered if every chain contains a smallest element.} set of configurations has a unique minimal equivalent, which is often called the Pareto frontier in other contexts. As a corollary of this theorem, every finite set of configurations has such a unique minimal equivalent.

\begin{thm}
If $C$ is a set of configurations, and $(C, \preceq^{-1})$ is well ordered, then there is a unique Pareto-minimal set of configurations $D$ such that $D \equiv C$.
\end{thm}

\begin{cor}
If $C$ is a finite set of configurations, then there is a unique Pareto-minimal set of configurations $D$ such that $D \equiv C$.
\end{cor}

\begin{defi}[Minimization]
Let $\mathcal{S}$ be a configuration space, and let $C \in \mathcal{S}$ such that $(C, \preceq^{-1})$ is well ordered. We let $\mathit{min}(C)$ denote the unique Pareto-minimal and Pareto-equivalent set of configurations.
\end{defi}

Minimization of a set of configurations can be done with, for instance, the Simple Cull algorithm (see, e.g.,~\cite{pareto-calc}) with a worst-case complexity of $\mathcal{O}(n^2)$ where $n$ is the number of configurations.
Below we define some operations on configuration spaces. 

Whether an operator on configuration spaces can safely use minimized operands is expressed through the notion of dominance preservation.

\begin{defi}[Dominance preservation]\label{def:sound}
Let $\mathcal{S}, \mathcal{T}$ be configuration spaces, let $f : \mathcal{S} \to \mathcal{T}$ be an operator with $n$ operands, and  let $C_1,\dots,C_n, C'_1,\dots,C'_n \subseteq \mathcal{S}$. We say that $f$ \emph{preserves dominance} if and only if
$\left( \forall_{1 \le i \le n}\,\, C_i \preceq C'_i \right) \Rightarrow f(C_1, \ldots, C_n) \preceq f(C'_1, \ldots, C'_n)$.
\end{defi}

\begin{cor}\label{cor:sound}
If $f$ preserves dominance, then $f$ can safely use Pareto-minimal operands, i.e.,
$f(C_1, \ldots, C_n) \equiv f(\mathit{min}(C_1), \ldots, \mathit{min}(C_n))$.
\end{cor}

Whether the result of an operator is already minimal is expressed by a notion of completeness that is formalized below.

\begin{defi}[Minimality preservation]\label{def:complete}
Let $\mathcal{S}, \mathcal{T}$ be configuration spaces, let $f : \mathcal{S} \to \mathcal{T}$ be an operator with $n$ operands, and  let $C_1,\dots,C_n \subseteq \mathcal{S}$. We say that $f$ \emph{preserves minimality} if and only if $\mathit{min}(f(C_1, \ldots, C_n)) = f(\mathit{min}(C_1), \ldots, \mathit{min}(C_n))$.
\end{defi}

Pareto-minimal sets of configurations are sufficient and necessary for the optimization of arbitrary cost functions which behave monotonically.

\begin{defi}[Derivation\footnote{This is a small and straightforward extension of the concept of \emph{monotonicity} from~\cite{pareto-algebra}.}]\label{def:derivations}
Let $\mathcal{S}$ be a configuration space, let $Q$ be a poset and let $f: \mathcal{S} \to Q$.
We say that $f$ is an $\preceq$-{derivation} (increasing derivation) if and only if $c \preceq c'$ implies that $f(c) \preceq_Q f(c')$ for all $c,c' \in \mathcal{S}$.
Similarly, we say that $f$ is a $\succeq$-{derivation} (decreasing derivation) if and only if $c \preceq c'$ implies that $f(c) \succeq_Q f(c')$ for all $c,c' \in \mathcal{S}$.
\end{defi}

\begin{defi}[Cost function]
Let $\mathcal{S}$ be a configuration space, let $C \subseteq \mathcal{S}$, let $Q$ be a poset and let $f: \mathcal{S} \to Q$. Then we say that $f$ is a cost function if and only if it is an $\preceq$-{derivation}.
We say that $c \in C$ is an optimum in $C$ if and only if $f(c') \preceq f(c)$ for all $c' \in C$, and $c$ is called optimal if and only if $f(c) \not\prec f(c')$ for all $c' \in C$.
\end{defi}

The following two theorems, taken from~\cite{pareto-algebra}, express (i) that dominated configurations are not needed to derive optimal configurations, and (ii) that every Pareto-minimal configuration in a configuration set is the unique optimum under at least some cost function.

\begin{thm}[Sufficiency of Pareto-minimal sets]\label{th:suf}
Let $\mathcal{S}$ be a configuration space, let $c, c' \in \mathcal{S}$, and let $f$ be a cost function. If $c' \preceq c$, then $f(c') \preceq f(c)$.
\end{thm}

\begin{thm}[Necessity of Pareto-minimal sets]\label{th:nec}
Let $\mathcal{S}$ be a configuration space, let $C \subseteq \mathcal{S}$ such that $C$ is Pareto minimal, and let $c \in C$.
Then there exists a cost function $f$ such that $c$ is the unique optimum configuration.
\end{thm}

Next, four operations on configuration spaces are defined that all preserve dominance and all but one preserve minimality.
We use these operations later on to define composition of components in our framework.

\begin{defi}[Free product]\label{def:fp}
Let $\mathcal{S}_1$ and $\mathcal{S}_2$ be configuration spaces, let $C_1 \subseteq \mathcal{S}_1$ and let $C_2 \subseteq \mathcal{S}_2$.
The free product of $C_1$ and $C_2$ is the Cartesian product $C_1 \times C_2$ in the configuration space $\mathcal{S}_1 \times \mathcal{S}_2$.
\end{defi}
\begin{prop}\label{prop:fcp-dom}
The free product preserves dominance.
\end{prop}
\begin{prop}\label{prop:fcp-min}
The free product preserves minimality.
\end{prop}

\begin{defi}[Safe constraint]\label{def:constraint}
Let  $\mathcal{S}$ be a configuration space. 
A constraint is a set $D \subseteq \mathcal{S}$.
A \emph{safe} constraint is a set $D \subseteq \mathcal{S}$ such that for all $c_1, c_2 \in \mathcal{S}$ holds that if $c_2 \preceq c_1$ and $c_2 \in D$, then $c_1 \in D$ (implying that dominating configurations cannot be excluded from a safe constraint).
Application of a constraint $D$ to a set of configurations $C \subseteq \mathcal{S}$ is defined as $C \cap D$.
\end{defi}
\begin{prop}\label{prop:con-dom}
Application of a safe constraint preserves dominance.
\end{prop}
\begin{prop}
Application of a safe constraint preserves minimality.
\end{prop}

\begin{defi}[Function application]\label{def:derivation}
Let $\mathcal{S}$ be a configuration space, let $C \subseteq \mathcal{S}$, let $Q$ be a poset and let $f: \mathcal{S} \to Q$.
Application of $f$ to $\mathcal{S}$ yields a new configuration space $\mathcal{S} \times Q$ and is defined as
$f (q_1, \ldots, q_n) = (q_1, \ldots, q_n, f(q_1, \ldots, q_n))$. We define $f(C) = \{ f(c) \,|\, c \in C \}$.
\end{defi}
Note that this definition of function application overloads the function name. Function name $f$ is also used to denote the application of $f$ to some configuration space.
\begin{prop}\label{prop:der-dom}
Application of an $\preceq$-{derivation} preserves dominance.
\end{prop}
\begin{prop}
Application of an $\preceq$-{derivation} preserves minimality.
\end{prop}

An application of a function can be written as an application of a constraint on the free product with the function range.
This is formalized by the following proposition.

\begin{prop}\label{prop:derivation-constraint}
Let $\mathcal{S}$ be a configuration space, let $Q$ be a poset, and let $f : \mathcal{S} \to Q$.
Then some constraint $D$ exists such that $f(C) = D \cap (C \times Q)$ (where $f(C)$ is function application as in Def.~\ref{def:derivation})  for all sets of configurations $C \subseteq \mathcal{S}$.
\end{prop}
\begin{proof}
We let $D = \{  (s_1, \ldots, s_n, f (s_1, \ldots, s_n)) \,|\, (s_1, \ldots, s_n) \in \mathcal{S} \}$.

($\subseteq$)  
Let $c=(s_1,\ldots,s_n, f(s_1,\ldots,s_n)) \in f(C)$.
By definition, $(s_1,\ldots,s_n) \in C \subseteq \mathcal{S}$ and $f(s_1,\ldots,s_n) \in Q$.
So clearly $c \in C \times Q$.
By definition, we also have that $c \in D$ and therefore $c \in D \cap (C \times Q)$.

($\supseteq$) 
Let $c=(s_1,\ldots,s_n, f(s_1,\ldots,s_n)) \in D \cap (C \times Q)$.
Thus, $(s_1,\ldots,s_n, f(s_1,\ldots,s_n)) \in (C \times Q)$ and therefore $(s_1,\ldots,s_n) \in C$.
By definition, $(s_1,\ldots,s_n, f(s_1,\ldots,s_n)) \in f(C)$.
\end{proof}

From this proposition, it follows that the application of any $\preceq$-derivation can also be written as the application of a constraint, constructed as in the proposition, to the free product of the input set and the range of the derivation.
This does not need to be a safe constraint. However, since application of a $\preceq$-derivation preserves dominance and minimality, this particular, equivalent, constraint application also preserves dominance and minimality.
An example is the function $f : \mathbb{N} \to \mathbb{N}$ (with $\le$ as order on $\mathbb{N}$) defined as $f(n) = 2$. Clearly, this is a $\preceq$-derivation, and has the constraint $D=\{ (n, 2) \,|\, n \in \mathbb{N} \}$. Although this is not a safe constraint, we can conclude that it nevertheless preserves dominance and minimality if used in $D \cap (C \times \mathbb{N})$ for all $C \subseteq \mathbb{N}$.

The dual of Prop.~\ref{prop:derivation-constraint}, which states that for every application of $D \cap (C \times Q)$ on a $C \subseteq \mathcal{S}$ we can find a function $f: \mathcal{S} \to Q$ that accomplishes the same, is not true (not even for safe constraints).
A counter example is $D = \emptyset$, which is safe, but which cannot be mimicked by some function since that function application can never result in an empty set if $C \neq \emptyset$.
Also note that $D \cap (C \times Q)$ can give a result $C'$ such that $|C| < |C'|$ (because of the use of the free product).
Again, this cannot be achieved with the derivation-based way.

\begin{defi}[Abstraction]\label{def:abstraction}
Let $\mathcal{S} = Q_1 \times Q_2 \times \ldots \times Q_n$ be a configuration space, and let $C \subseteq \mathcal{S}$. The $k$-abstraction is a function $\mathcal{S} \to Q_1 \times \ldots \times Q_{k-1} \times Q_{k+1} \times \ldots \times Q_n$ defined as $ (q_1, q_2, \ldots, q_n) \downarrow {k} = (q_1, \ldots, q_{k-1}, q_{k+1}, \ldots, q_n)$. We define $C \downarrow {k} = \{ c \downarrow k \,|\, c \in C \}$.
\end{defi}
\begin{prop}\label{prop:abs-dom}
Application of $k$-abstraction preserves dominance.
\end{prop}

Abstraction does not preserve minimality because an abstraction can remove the single dimension that made some configuration incomparable to any other configuration. For instance, take $\mathcal{S} = \mathbb{N} \times \mathbb{N}$ with $\le$ as ordering on $\mathbb{N}$. Then $C = \{ (1,2), (2,1) \}$ is Pareto minimal. We have $C \downarrow 1 = \{ 2, 1\}$, which clearly is not Pareto minimal because $1 < 2$.

The following is not explicit in~\cite{pareto-algebra} but follows trivially.
Permutation allows us to reorder the dimensions of a configuration space.

\begin{defi}[Permutation]\label{def:perm}
Let $ \mathcal{S} = Q_1 \times Q_2 \times \ldots \times Q_n$ be a configuration space, let $C \subseteq \mathcal{S}$ be a set of configurations, and let $\pi : \{1, 2, \ldots, n\} \to \{1, 2, \ldots, n\}$ be a bijection, i.e., a permutation.
We define $\pi(\mathcal{S}) = Q_{\pi(1)} \times Q_{\pi(2)} \times \ldots \times Q_{\pi(n)}$
and also let $\pi$ denote the function $\mathcal{S} \to \pi(\mathcal{S})$ defined as
$\pi(q_1, q_2, \ldots, q_n) = (q_{\pi(1)}, q_{\pi(2)}, \ldots, q_{\pi(n)})$.
Furthermore, we define $\pi(C) = \{ \pi(c) \, | \, c \in C \}$.
\end{defi}
\begin{prop}\label{prop:perm-dom}
Permutation preserves dominance.
\end{prop}
\begin{prop}
Permutation preserves minimality.
\end{prop}

Finally, we need one more concept from Pareto Algebra that we use in our quality- and resource-management specialization later on.

\begin{defi}[Alternatives]\label{def:union}
Let $ \mathcal{S}$ be a configuration space, and let $C_1, C_2 \subseteq \mathcal{S}$.
Then $C_1 \cup C_2$ is called the set of alternatives of $C_1$ and $C_2$.
\end{defi}
\begin{prop}\label{prop:union-dom}
Alternatives preserve dominance.
\end{prop}

The alternatives operation does not preserve minimality.
For instance, take $\mathcal{S} = \mathbb{N}$ with $\le$ as ordering on $\mathbb{N}$. Then both $C_1 = \{ 1 \}$ and $C_2 = \{2 \}$ are Pareto minimal. However, $C_1 \cup C_2 = \{ 1, 2\}$ clearly is not Pareto minimal because $1 < 2$.


\section{An Interface for Quality and Resource Management}\label{sec:components}
The information needed for quality- and resource-management purposes is captured in our framework by an interface model for components.
A component has an input, an output, a required budget, a provided budget, a quality, and it may be parameterized.
The input, required budget, and parameters model the {input assumptions}, and the output, provided budget and quality model the {output guarantees} of the component~\cite{AH01b}.
All six parts are modeled by posets.
Each of these posets can be multi-dimensional, to capture cases with multiple inputs, outputs, resource requirements, etc.
This is the \emph{quality- and resource-management interface}  of the component.

\begin{defi}[QRM interface]\label{def:qrmi}
The QRM interface of a {component} is a set of configurations from a six-dimensional configuration space $Q_i \times Q_o \times Q_r \times Q_p \times Q_q \times Q_x$.
By convention, the first dimension $Q_i$ models the input of the component, the second dimension $Q_o$ models the output, the third dimension $Q_r$ models the required budget,  the fourth dimension $Q_p$ models the provided budget, the fifth dimension $Q_q$ models the quality, and the sixth dimension $Q_x$ models the parameters.
Different components can use different configuration spaces.
\end{defi}

A component can thus have a QRM interface that has multiple configurations. Each configuration models a specific working point of the component with specific values for input, output, required budget, provided budget, quality and parameters.
These values usually are the result of some kind of analysis that is applied to the technology that is used to implement the component.
External actors, e.g., end users, can typically control selection of subsets of configurations through the parameters. The freedom that remains to choose the final configuration is left to the quality and resource manager for optimization.

Figure~\ref{fig:component} shows the graphical format that we use for a single configuration $c$ of the QRM interface  $C \subseteq Q_i \times Q_o \times Q_r \times Q_p \times Q_q \times Q_x$ of some component.
Depicting the complete interface with all configurations is often problematic as the values of input $i$, output $o$, required budget $r$, provided budget $p$, quality $q$, and parameters $x$ differ between configurations.
The interpretation of the partial-order relation on input is that if $i \preceq i'$, then $i'$ is considered \emph{weaker}\footnote{For convenience we do not distinguish between the strict and non-strict versions for the description of the order relations.}.
For output, we interpret $o \preceq o'$ as that $o'$ is \emph{stronger} than $o$.
For the required budget, we say that $r \preceq r'$ means that $r'$ is \emph{smaller} than $r$.
For the provided budget, we say that $p \preceq p'$ means that $p'$ is \emph{larger} than $p$.
The quality dimension is interpreted as follows: $q \preceq q'$ means that $q'$ is \emph{better} than $q$.
Parameters are typically unordered; there are no better or worse parameter values. Whether a  configuration is better or worse than, or incomparable to, another configuration is determined by the other five parts of the configurations.

\begin{figure}
\centering
\includegraphics[width=0.5\linewidth]{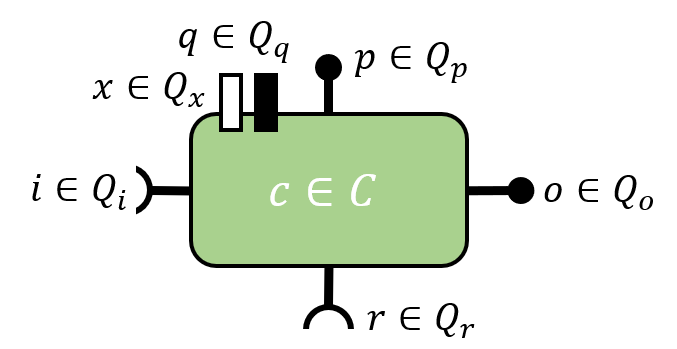}
\caption{The interface for quality and resource management.}\label{fig:component}
\end{figure}

In our framework, input- and required-budget specifications capture {requirements}, and output, provided budget and quality capture promises.
Requiring more is worse than requiring less. We have diametrically opposed interpretations of $\preceq$ for input and output, and also for required budget and provided budget, to be able to use the dominance relation of Def.~\ref{def:dominance}.
The dominating configuration thus has a weaker input requirement, a stronger output guarantee, a smaller required budget, a larger provided budget, a better quality and better, i.e., equal, parameters.
This \emph{alternating} approach to dominance is similar to the notion of alternating simulation that is used to define refinement for interface automata~\cite{AH01}, and is also used, for instance, in contract-based design to define refinement between contracts~\cite{B08}.
Note that if we have a poset $Q_{io}$ with partial order $\preceq_i$ that models the input and we need a poset for modeling the output that has the same semantic domain, we can just take the dual poset, i.e.,  the same set $Q_{io}$ with partial order $\preceq_o$ where $x \preceq_o y$ if and only if $y \preceq_i x$ (i.e., $\preceq_o{}={}\preceq_i^{-1}$). Clearly, the same goes for the required and provided budgets.

\begin{exa}\label{ex:2}
Consider Ex.~\ref{ex:1}. The posets used to model the component parts are shown at the bottom of the figure.
Now let us look at three configurations of the Scaler component, that scale FHS@60 input video to different resolutions and rates. Formally, these are the following\footnote{The Scaler component (and other components also) have many more configurations than mentioned in the examples. We only show some of these configurations to make our examples concise.}:
\[
\begin{array}{rccccccl}
& \mathrm{input} & \mathrm{output}& \mathrm{required}& \mathrm{provided}& \mathrm{quality}& \mathrm{parameters} \\
& (\mathbb{N}^3, =) & (\mathbb{N}^3, =) & (\mathbb{N}^2, \ge^2) & (\{ \bot \}, =) & (\mathbb{N}, \le) & (\mathbb{N}^2, =) \\
\hline
s_1 = \big(& \mathrm{FHD@60}, & \mathrm{{HD+}@60}, & (201,15), & \bot, & 60, & \mathrm{HD+} & \big) \\
s_2 = \big(& \mathrm{FHD@60}, & \mathrm{{HD+}@30}, & (160,15), & \bot, & 30, & \mathrm{HD+} &\big) \\
s_3 = \big(& \mathrm{FHD@60}, & \mathrm{HD@60}, & (171,15), & \bot, & 60, & \mathrm{HD} &\big) \\
\end{array}
\]
Here we use the  \emph{void poset} $\{ \bot \}$ with partial order $=$ to model the provided budget of the Scaler.
Note that the parameters poset uses equality as its ordering relation as parameters are unordered.
Therefore, $s_3$ neither dominates nor is dominated by either $s_1$ or $s_2$ because it has a different parameter value.
Because equality is also used for the output, $s_1$ neither dominates nor is dominated by $s_2$.
Now suppose that the Scaler has two other configurations that employ a different scaling algorithm that needs more buffering space:
\[
\begin{array}{rccccccl}
s_4 = \big(& \mathrm{FHD@60}, & \mathrm{{HD+}@60}, & (201,30), & \bot, & 60, & \mathrm{HD+} &\big) \\
s_5 = \big(& \mathrm{FHD@60}, & \mathrm{{HD+}@30}, & (160,30), & \bot, & 30, & \mathrm{HD+} & \big)
\end{array}
\]
Then we have that $s_4 \preceq s_1$ and $s_5 \preceq s_2$ because the different scaling algorithm requires a larger budget but does not provide anything in terms of quality or output.
Thus, $C = \{ s_1, \ldots, s_5 \}$, is not Pareto minimal. Its Pareto-minimal and equivalent component is $D = \{ s_1, s_2, s_3 \}$, because $C \equiv D$, and  $D$ is Pareto minimal because $\{s_1, s_2, s_3\}$ forms an anti-chain with respect to the dominance relation.
Adding the algorithm as an actual parameter would result in $\{ s_1, \ldots, s_5 \}$ being Pareto minimal.
Also the situation where $s_4$ and $s_5$ have an additional quality such as using less energy, would have this result.

The parameters can be used for several purposes. In this example, the parameters can be used to control the output resolution of the Scalers.
We mentioned that parameters also can be used to differentiate between configurations that have different implementation algorithms. Such implementation alternatives may or may not be equal in terms of the input, output, budgets and quality. A related example is that they can be used to differentiate between different versions of components, where the version can be captured as an unordered quality to allow on-line version management.
\end{exa}

A quality- and resource-management challenge arises when there are multiple components that implement the QRM interface. All these components must be combined in a proper way while optimizing some objective(s).
In the next section, we therefore explain how Pareto-algebraic concepts can be used to create compositions of components that implement the QRM interface. These compositions can again be described by a QRM interface. The whole composition process is thus expressed in Pareto-algebraic terms, and the notion of Pareto minimality applies. This naturally gives the necessary tools for multi-objective optimization.


\section{Composition of QRM interfaces}\label{sec:composition}
We see composition as an \emph{aggregation of alternatives}, and realize it with the Pareto-algebraic operations introduced in Sec.~\ref{sec:pa}. The basic idea is that component interfaces can be realized with alternative options that are then aggregated in compositions as already illustrated in Fig.~\ref{fig:video}.
The \emph{alternatives pattern} and \emph{aggregation pattern} are defined in the following subsections.

\subsection{Alternatives pattern}
We first generalize the binary alternatives operation of Def.~\ref{def:union} to an $n$-ary operation $\cup^n$.
Note that the alternatives operation requires that the interfaces are from the same configuration space. 
Therefore, to apply the alternatives operation to interfaces
$C_1 \subseteq \mathcal{S}_1, \ldots, C_n \subseteq \mathcal{S}_n$, we must first \emph{normalize}
$C_1, \ldots, C_n$ to $C_1', \ldots, C_n'$ respectively, such that $C_1' \subseteq \mathcal{S}', \ldots, C_n' \subseteq \mathcal{S}'$ for some configuration space $\mathcal{S}'$.
This normalization is done through proper derivations for each part of the interface, followed by abstraction and minimization. The pattern for taking the alternatives of $n$ interfaces is the following:
\begin{equation}\label{eq:alt}
\begin{array}{rl}
(1) & \mathit{min} \circ (\downarrow 1)^{6} \circ f^1_x \circ f^1_q \circ f^1_p \circ f^1_r \circ f^1_o \circ f^1_i \\
\vdots\,\,\, \\
(n) & \mathit{min} \circ (\downarrow 1)^{6} \circ f^n_x \circ f^n_q \circ f^n_p \circ f^n_r \circ f^n_o \circ f^n_i \\
(n+1) & \mathit{min} \circ \cup^n
\end{array}
\end{equation}

Here the $f_x^1, \ldots, f_i^n$ are the $\preceq$-derivations that normalize the input, output, etcetera.
Since these derivations add dimensions to the configuration space, we need to abstract from the original dimensions to get back to a six-dimensional configuration space with the QRM interpretation.
The first step normalizes the first interface, the second step normalizes the second interface,
etcetera. The $n+1$-th step applies the alternatives operation to all normalized interfaces and minimizes the result.

The following proposition states that we can write the normalization as a constraint applied to a free product with the normalized configuration space.

\begin{prop}\label{prop:rewrite}
Let $\mathcal{S}$ be a configuration space, and let $f = f_x \circ f_q \circ f_p \circ f_r \circ f_o \circ f_i$ with $f_i : \mathcal{S} \to Q_i$, $f_o : \mathcal{S} \times Q_i \to Q_o$, etcetera.
Then, some constraint $D$ exists such that $f(C) = D \cap(C \times Q_i \times Q_o \times Q_r \times Q_p \times Q_q \times Q_x)$ for all $C \subseteq \mathcal{S}$.
\end{prop}
\begin{proof} Let $D_i$, $D_o$, etcetera, be the constraints as defined in Prop.\ \ref{prop:derivation-constraint}, for $f_i$, $f_o$, etcetera.
\[
\begin{array}{rl}
f_x \circ f_q \circ f_p \circ f_r \circ f_o \circ f_i(C) \\
	= & \mathrm{Prop.~\ref{prop:derivation-constraint}} \\
f_x \circ f_q \circ f_p \circ f_r \circ f_o \big(  D_i \cap (C \times Q_i ) \big) \\
	= & \mathrm{Prop.~\ref{prop:derivation-constraint}} \\
f_x \circ f_q \circ f_p \circ f_r \Big(  D_o \cap \big(  ( \underbrace{D_i}_{X} \cap \underbrace{(C \times Q_i )}_{Y} ) \times \underbrace{Q_o}_{Z}  \big)  \Big) 
\\
	= & (X \cap Y) \times Z = (X \times Z) \cap (Y \times Z) \\
f_x \circ f_q \circ f_p \circ f_r \Big(  D_o \cap \big( (D_i \times Q_o) \cap (C \times Q_i \times Q_o )  \big)  \Big) \\
	= & X \cap (Y \cap Z) = (X \cap Y) \cap Z \\
	
f_x \circ f_q \circ f_p \circ f_r \Big(  \big( D_o \cap  (D_i \times Q_o) \big) \cap (C \times Q_i \times Q_o )  \Big)  \\
	= &  D_o \cap  (D_i \times Q_o) = D'  \\

f_x \circ f_q \circ f_p \circ f_r \big( D'  \cap (C \times Q_i \times Q_o )  \big)  \\
	= & \ldots \\

D  \cap (C \times Q_i \times Q_o \times Q_r \times Q_p \times Q_q \times Q_x)
\end{array}
\]
  \par\vspace{-1.3\baselineskip}\qedhere
\end{proof}

This proposition thus gives us a second way to accomplish normalization based on constraining the free product with the normalized configuration space $Q_i \times Q_o \times Q_r \times Q_p \times Q_q \times Q_x$.
This results in the following constraint-based formulation of the alternatives pattern:
\begin{equation}\label{eq:alt2}
\begin{array}{rl}
(1) & \mathit{min} \circ (\downarrow 1)^{6} \circ\cap D_1 \circ \times({}\cdot{}, Q_i \times Q_o \times Q_r \times Q_p \times Q_q \times Q_x)  \\
\vdots\,\,\, \\
(n) & \mathit{min} \circ (\downarrow 1)^{6} \circ\cap D_n \circ \times({}\cdot{}, Q_i \times Q_o \times Q_r \times Q_p \times Q_q \times Q_x)  \\
(n+1) & \mathit{min} \circ \cup^n
\end{array}
\end{equation}

The proposition above shows that any normalization with derivations as in Eq.~\ref{eq:alt} can be written using constraints as in Eq.~\ref{eq:alt2}. The reverse is not true, i.e., the constraint-based way is strictly more expressive than the derivation-based way. In line with the reasoning below Prop.~\ref{prop:derivation-constraint},  constraint-based normalization applied to a set of configurations $C$ can result in the empty set (if the constraint is empty) and it may give a result $C'$ such that $|C| < |C'|$. Both cannot be achieved with the derivation-based pattern.

Although the constraint-based alternatives pattern is more expressive, it can be easier to reason about the preservation of dominance and minimality - which is crucial in the Pareto Algebra framework - with the derivation-based pattern.  The example below Prop.~\ref{prop:derivation-constraint} shows that a simple $\preceq$-derivation cannot be mimicked with a safe constraint. As we see below however, this kind of derivation is useful for normalization.
Only using safe constraints in Eq.~\ref{eq:alt2} thus might limit the practical applicability.
This consideration is the reason why we present both Eq.~\ref{eq:alt} and Eq.~\ref{eq:alt2} as patterns for taking alternatives of components.

\subsection{Aggregation pattern}
Aggregation combines $n$ constituent QRM interfaces into a new QRM interface. It consists of several steps and starts from a straightforward generalization of the free product in Def.~\ref{def:fp} that combines $n$ interfaces into a $6n$-dimensional configuration space.
Second, we apply constraints originating from, e.g., i/o and budget matching, and constraints on allowed parameter combinations to restrict the composition to feasible configurations.
Third, we apply operations that express the semantics of the aggregation in the application domain.
We assume that we have an aggregation operator for all six parts (input, output, etcetera) of the interface.
Typically such an aggregation operator takes the parts of the constituent interfaces and creates a new part for the aggregation using a number of building-block functions, most notably domain-specific $\preceq$-derivations that can be applied using Def.~\ref{def:derivation}.
Fourth, after application of the proper derivations, we apply any relevant constraints on the resulting configurations.
Fifth, we reconcile the result with our notion of QRM interface by abstraction of dimensions that represent the constituent interfaces. This results again in a six-dimensional configuration space of the aggregation with the usual interpretation of the dimensions. 
Finally, we apply Pareto minimization. All this is expressed in the following template:
\begin{equation}\label{eq:aggr}
\mathit{min} \circ (\downarrow 1)^{6n} \circ \cap D_a \circ f_x \circ f_q \circ f_p \circ f_r \circ f_o \circ f_i \circ \cap D_c \circ \times^n
\end{equation}
\noindent
Here $\times^n$ is the generalized free product, and $f_i, \ldots$ are the aggregation functions for input, output, etcetera.
The constraints are split into two parts: $D_c$, applied directly after the free product on the constituent parts, and $D_a$, applied after the derivations that compute the interface values of the aggregation.
Typically, i/o and budget matching, and parameter constraints are part of $D_c$. 
Domain-specific quality constraints, e.g., a constraint on a quality such as output rate, or constraints on the parameters of the composition can be part of $D_a$.
Finally, $(\downarrow 1)^{6n}$ abstracts the dimensions of the constituent interfaces, and $\mathit{min}$ is the minimization operation.

Note that an actual implementation of this aggregation can be done much more efficiently than what is shown here.
Typically, the application of constraints and the derivations and abstractions can be inlined during the creation of the free product.
Since abstraction is always used in our composition and it does not preserve minimality, we always need to recalculate the set of Pareto-minimal configurations. This can be done using, e.g., the Simple Cull algorithm with quadratic worst-case complexity.
Computation of the free product with the inline calculation of the constraints and derivations, can on its turn also be inlined in the Simple Cull algorithm, making it suitable for on-line operation~\cite{pareto-calc}.

Eq.~\ref{eq:aggr} can be rewritten to the following form in which the derivations are replaced by a constraint in a similar way as is shown for the alternatives pattern above:
\begin{equation}\label{eq:aggr2}
\mathit{min} \circ (\downarrow 1)^{6n} \circ \cap D \circ \times({}\cdot{}, Q_i \times Q_o \times Q_r \times Q_p \times Q_q \times Q_x) \circ \times^n
\end{equation}
This pattern combines all the constraints in $D_c$ and $D_a$ with those originating from the aggregation functions in a single constraint $D$.

Composition of a number of components is accomplished by repeated application of the alternatives and aggregation patterns of Eq.~\ref{eq:alt}, Eq.~\ref{eq:alt2}, Eq.~\ref{eq:aggr}, and Eq.~\ref{eq:aggr2}:
composition is an aggregation of alternatives.
If we use dominance-preserving constraints in these operations, then the total composition function also preserves dominance since the other Pareto-algebraic operations used in the operation preserve dominance.
Furthermore, since all operations end with a minimization, the total composition function also preserves minimality.

\subsection{Refinement}
The next theorem formalizes refinement in our setting. If component implementations are at least as good as is stated in their quality and resource management interfaces, then the composition of the implementations is at least as good as the composition based on the interfaces.

\begin{thm}[Refinement]\label{th:refinement}
Let $C_1 \subseteq \mathcal{S}_1, \ldots, C_n \subseteq \mathcal{S}_n$ and $C_1' \subseteq \mathcal{S}_1, \ldots, C_n' \subseteq \mathcal{S}_n$ be components, let $C_i \preceq C_i'$ for all $1 \le i \le n$, and
let $f: 2^{\mathcal{S}_1} \times \ldots \times 2^{\mathcal{S}_n} \to 2^\mathcal{S}$ be a function composition that only uses the basic Pareto-algebraic operations (free product, constraint application with safe constraints, derivation, abstraction, and alternatives, e.g., an aggregation of alternatives).
Then, $f(C_1, \ldots, C_n) \preceq f(C_1', \ldots, C_n')$.
\end{thm}
\begin{proof}
By induction on the number of the basic Pareto-algebraic operations that we use in $f$.
The base case for a single operation is proven by dominance preservation of the basic operators, see Props.~\ref{prop:fcp-dom}, \ref{prop:con-dom}, \ref{prop:der-dom}, \ref{prop:abs-dom},  \ref{prop:union-dom}.
Now suppose that the theorem holds for $k-1$ operations.
Then we have composed an intermediate configuration set $D$ from the set $C_1,\ldots,C_n$, and we can compose an intermediate configuration set $D'$ from the set $C_1',\ldots,C_n'$ such that $D \preceq D'$.
Applying the next operation to $D$ and $D'$ results in $E$ and $E'$ respectively.
By dominance preservation of the composition operators and the induction hypothesis that $D \preceq D'$, we have that $E \preceq E'$. 
\end{proof}

A component declares a QRM interface consisting of a number of configurations $C$ using Def.\ \ref{def:qrmi}. This is an abstraction of the real configurations of the component, the implementation, that can also be expressed as a QRM interface $C'$, but which may not be exactly known.
The assumption that we make, however, is that the implementation is at least as good as declared. That is to say that the implementation dominates the interface: $C \preceq C'$.
Note that this allows the implementation to have more, but also less configurations than declared in the interface.

The theorem above is defined in terms of dominance of sets of configurations without taking a particular cost function into account.
Now suppose that we have created a composition $f$ based on the interface models and using a certain cost function we find an optimal configuration of this composition.
I.e., we have $f(C_1, \ldots, C_n) = C$ and the cost function selects $c \in C$ as optimum.
We can trace for each interface which configuration was used to achieve this optimum: let $c_i \in C_i$ be the configurations used to obtain $c$.
The question then is what happens when in an implementation, components and configurations are used that are refinements of the interface models, leading to the realization of the configuration $c'$ instead of $c$?
According to the definition of set dominance, for every interface configuration $c_i \in C_i$, there exists an implementation configuration $c_i'$ in the implementation $C_i'$ of the component, that is at least as good: $c_i \preceq c_i'$.
From Th.~\ref{th:refinement}, it follows that there exists a configuration $c'$ in the composition of the implementations such that $c \preceq c'$. This thus means that an optimization done based on the interfaces is \emph{safe} in the sense that the corresponding implementations yield a result that is at least as good for any arbitrary cost function.

The monotonic nature of our framework fits well with other formalisms with a monotonic nature that may be used to model and analyze the relation between the various quality and resource management aspects of the component configurations.
An example is the synchronous dataflow (SDF) formalism~\cite{sdf} that is suitable for analyzing streaming applications such as the video-processing system of our motivating example.
It is well-known that this formalism is monotone in the sense that reducing the execution times of actors\footnote{An actor is part of the SDF formalism and models some (computational) task with a fixed execution time.} will not result in a decrease of throughput. This is used in \cite{dfrefine} for compositional refinement that preserves key real-time performance metrics such as throughput and latency.
 We can also use the monotonicity in our setting.
Suppose that some component can be modeled using SDF, and that its quality is throughput and its required budget relates directly  to the execution times of the actors.
SDF analysis gives concrete values for several component configurations, i.e., the analysis calculates the throughput for a number of required budget values (actor execution times).
Now let us consider the composition with another component that provides the budget for the actor execution.
The monotonicity of SDF in combination with Thm.~\ref{th:refinement} provides several robustness benefits.
First, if the provided budget is larger than the required budget, then the throughput will not decrease because of SDF monotonicity. This means that the partial order on the actor execution budget can be a more flexible partial order than just equality.
This eases implementation of the providing component: not exactly the required budget should be delivered, but a slightly larger budget (just to be safe) is also good according to Thm.~\ref{th:refinement}.
Second, usually the worst-case execution time of actors is used during SDF analysis. Because of monotonicity of SDF, an implementation will not have worse throughput. In terms of our framework, the component implementation dominates the interface that was used for composition.
By Thm.~\ref{th:refinement}, this is no problem and the result is at least as good as was predicted based on the interface.


\section{Derivations and constraints for composition}\label{sec:recipe2}
In this section, we detail the use of Eq.~\ref{eq:alt} and Eq.~\ref{eq:aggr} for composition, which we see as an aggregation of alternatives, using the running example.
The use of the constraint-based way for composition, as defined by Eq.~\ref{eq:alt2} and Eq.~\ref{eq:aggr2}, is used and further discussed in Sec.~\ref{sec:qrml}, where we introduce a domain-specific language (DSL) for specifying the QRM view of a system with its semantics in terms of the framework developed in this paper.

\subsection{Alternatives}
The description of our running example in Sec.\ \ref{sec:ex} does not contain alternatives.
A variation of the case, however, also considers a software implementation of the scaler functionality that performs the processing.
Remember that the hardware scaler provides a three-fold budget to model the supported number of streams, the computational power, and the available memory segments respectively.
The software scaler can support a virtually unlimited number of streams, has virtually unlimited memory resources, but has less computational power. It therefore only provides computation in its provided budget.
In the following example, we show how we can create a component that models the choice between a hardware scaler and a software scaler.

\begin{exa}[Alternatives]\label{ex:alt}
The hardware scaler has configuration space $V \times V \times V \times \mathbb{N}^3 \times V \times V$ where $V = \{ \bot \}$ with order $=$ is the void poset. Its only configuration is the following:
\[
h = (\bot, \bot, \bot, (4, 300, 32), \bot, \bot)
\]
The software scaler has configuration space $V \times V \times V \times \mathbb{N} \times V \times V$ and also has a single configuration:
\[
s = (\bot, \bot, \bot, 100, \bot, \bot)
\]
To model the choice between these two components, we need to normalize them to a single configuration space.
We choose the space $V \times V \times V \times (\mathbb{N} \cup \{\top\})^3 \times V \times V$.
The special $\top$ element is larger than any element of $\mathbb{N}$ and thus models an unlimited provided budget.

Normalization of the provided budget of the hardware scaler is now modeled by the function $f_p^1 : \mathbb{N}^3 \to (\mathbb{N} \cup \{\top\})^3$ defined as $f_p^1(x,y,z) = (x,y,z)$ (i.e., it need not be changed).
Normalization of the provided budget of the software scaler is modeled by the function $f_p^2 : \mathbb{N} \to (\mathbb{N} \cup \{\top\})^3$ defined as $f_p^2(y) = (\top,y,\top)$ (the unlimited budgets are added).
The normalization functions for all parts but the provided budget are just copy functions $\mathit{copy}_i$ that copy dimension $i$ (see below for a precise definition). These copy functions trivially are $\preceq$-derivations.
Together, this gives the following instance of  Eq.~\ref{eq:alt}:
\[
\begin{array}{rl}
(1) & \mathit{min} \circ (\downarrow 1)^{6} \circ \mathit{copy}_6 \circ \mathit{copy}_5 \circ f^1_p \circ \mathit{copy}_3 \circ\mathit{copy}_2 \circ\mathit{copy}_1 \\
(2) & \mathit{min} \circ (\downarrow 1)^{6} \circ \mathit{copy}_6 \circ\mathit{copy}_5 \circ f^2_p \circ\mathit{copy}_3 \circ \mathit{copy}_2 \circ \mathit{copy}_1 \\
(3) & \mathit{min} \circ \cup
\end{array}
\]
\end{exa}
After application of the alternatives operation to $\{h\}$ and $\{s\}$, we obtain a general HWorSWscaler interface for the execution platform with two 
configurations:

\[
\{ (\bot, \bot, \bot, (4, 300, 32), \bot, \bot), (\bot, \bot, \bot, (\top, 100, \top), \bot, \bot) \}
\]

\subsection{Aggregation}
As a first example of an aggregation, we show a free aggregation without any i/o or budget matching. This can, for example, be used to compose a platform from Fiber and HW scaler components.
Some simple and generally useful $\preceq$-derivations that can be used to create the functions $f_i$, $f_o$, etcetera are defined next.
First, a grouping derivation can be defined as follows that puts dimensions $i$ to $j$ in a single new multi-dimensional dimension. This new dimension uses the element-wise $\preceq$ relation of its subdimensions (i.e., the dominance relation).

\begin{defi}[Grouping]\label{def:group}
Let $\mathcal{S} = Q_1 \times \ldots \times Q_n$ be a configuration space and let $1 \le i  \le j \le n$.
We define $\mathit{group}_{i,j} : \mathcal{S} \to Q$ with  $Q = Q_i \times \ldots \times Q_j$ and $(p_i, \ldots, p_j) \preceq_Q (q_i, \ldots, q_j)$ if and only if $p_k \preceq_{Q_k} q_k$ for all $i \le k \le j$ as follows:
$\mathit{group}_{i,j}(q_1, \ldots, q_n) = (q_i, \ldots, q_j)$.
\end{defi}
Note that $\mathit{group}_{i,i}$ just copies a dimension. We abbreviate $\mathit{group_{i,i}}$ by $\mathit{copy}_i$.

\begin{prop}
Grouping is an $\preceq$-derivation. 
\end{prop}
\begin{proof}Directly from Def.\ \ref{def:derivations}.\end{proof}

Next, ungrouping can help to access parts of multi-dimensional posets.

\begin{defi}[Ungrouping]
Let $\mathcal{S} = Q_1 \times \ldots \times Q_n$ be a configuration space, let $1 \le i \le n$, let $Q_i = Q_i^1 \times \ldots \times Q_i^m$ and let $1 \le j \le m$.
Furthermore, let $(p_i^1, \ldots, p_i^m) \preceq_{Q_i} (q_i^1, \ldots, q_i^m)$ if and only if $p_i^k \preceq q_i^k$ for all $1 \le k \le m$.
We define $\mathit{ungroup}_{i,j} : \mathcal{S} \to Q_i^j$  as  $\mathit{ungroup}_{i,j}(q_1, \ldots, q_{i-1},  (q_i^1, \ldots, q_i^m), q_{i+1}, \ldots, q_n) = q_i^j$.
\end{defi}
\begin{prop}
Ungrouping is an $\preceq$-derivation. 
\end{prop}
\begin{proof}Directly from Def.\ \ref{def:derivations}.\end{proof}

Finally, the void derivation adds a void value that can be used, e.g., to model empty component parts.

\begin{defi}[Voiding]
Let $\mathcal{S}$ be a configuration space.
We define $\mathit{void} : \mathcal{S} \to \{ \bot \}$ as  $\mathit{void}(c) = \bot$, with $=$ as the order on $\{ \bot \}$.
\end{defi}
\begin{prop}
Voiding is an $\preceq$-derivation.
\end{prop}
\begin{proof}Directly from Def.\ \ref{def:derivations}.\end{proof}
The next example shows how we apply all this for a simple aggregation.

\begin{exa}[Free aggregation]\label{ex:comp1}
In this example, we aggregate a Fiber with a HW scaler, both introduced in Sec.~\ref{sec:ex}.
Since there is no connection made between input and output, or between provided and required budget, we call this a \emph{free aggregation}.
We assume that we have a single Fiber that provides 10~Gb/s bandwidth and a single HW scaler that provides a scalers budget of four streams, 300~Mpixels/s processing, and 32 memory segments of 128 pixels each.
Both components have a single configuration, and only a non-void provided budget:
\[
\begin{array}{rccccccl}
f_1 = (& \bot, & \bot, & \bot, & 10, & \bot, & \bot &) \\
h_1 =(& \bot, & \bot, & \bot, & (4, 300, 32), & \bot, & \bot &)
\end{array}
\]
We then apply the free product to obtain:
\[
\{f_1\} \times \{h_1\} = \{ ( \overbrace{\bot, \bot, \bot, 10, \bot, \bot}^{f_1}, \overbrace{\bot, \bot, \bot, (4, 300, 32), \bot, \bot}^{h_1} ) \}
\]

Next, we define the compositions for the six interface parts. The composition of input, output and required budget is modeled by the void derivation, because the Fiber and HW scaler do not have these parts.
The composition of the provided budget is done by grouping the respective provided budgets. Finally, we again use the void derivation for the composition of quality and parameters.
This gives the following configuration:
\[
( \overbrace{\bot, \bot, \bot, 10, \bot, \bot}^{f_1},
\overbrace{\bot, \bot, \bot, (4, 300, 32), \bot, \bot}^{h_1},
\bot, \bot, \bot, (10, (4, 300, 32)), \bot, \bot
 )
\]

The final steps are to abstract from the individual component parts, i.e., the first twelve dimensions, and to minimize (which is trivial for a configuration set of size one). This yields the execution platform component with only a single configuration.
\[
e_1 =  (\bot, \bot, \bot, (10, (4, 300, 32)), \bot \bot )
\]
This aggregation can thus be written as the following instance of Eq.~\ref{eq:aggr} (where constraints are not present, as mentioned before):
\[
\mathit{min} \circ (\downarrow 1)^{12} \circ \mathit{void} \circ \mathit{void} \circ\mathit{group}_{4,10} \circ \mathit{void} \circ \mathit{void} \circ \mathit{void} \circ \times
\]
\end{exa}

Our next example shows a more complicated aggregation. We want to make sure that input and output match. We do this by application of a producer-consumer constraint, which is already introduced in ~\cite{pareto-algebra}.

\begin{defi}[Producer-consumer constraint]\label{def:pc}
Let $\mathcal{S} = Q_1 \times Q_2 \times \ldots \times Q_n$ be a configuration space, let $1 \le p \neq c \le n$ and let $f : Q_c \to Q_p$ be a $\succeq$-derivation. Then, $D_{p,c,f} = \{ (q_1, q_2, \ldots, q_n) \in \mathcal{S} \,|\, f(q_c) \preceq_{Q_p} q_p \}$ is a producer-consumer constraint.
\end{defi}
\begin{prop}
A producer-consumer constraint is safe.
\end{prop}
\begin{proof}
Let $d = (q_1, \ldots, q_n), d'= (q_1', \ldots, q_n') \in \mathcal{S}$ and let $d \preceq d'$. Let us assume that $d \in D_{p,c,f}$. We have to show that $d' \in D_{p,c,f}$, which is to say that $f(q_c') \preceq_{Q_p} q_p'$.
We have that $f(q_c) \preceq_{Q_p} q_p$ and since $d \preceq d'$ also that $q_c \preceq q_c'$ and $q_p \preceq q_p'$.
Since $f$ is a $\succeq$-derivation, we have that $f(q_c) \succeq_{Q_p} f(q_c')$.
Thus, $f(q_c') \preceq_{Q_p} f(q_c) \preceq_{Q_p} q_p \preceq_{Q_p} q_p'$.
\end{proof}

Another generally applicable constraint lets us select an arbitrary subset of configurations based on an unordered dimension.
This is typically used to select arbitrary parameter valuations.

\begin{defi}[Subset constraint]
Let $\mathcal{S}= Q_1 \times Q_2 \times \ldots \times Q_n$ be a configuration space,  let $1 \le i \le n$, let $=$ be the partial order for $Q_i$, and let $X \subseteq Q_i$.
The constraint $D_{i,X} = \left\{ (q_1, q_2, \ldots, q_n) \in \mathcal{S} \,|\,  q_i \in X \right\}$ is called a subset constraint.
\end{defi}
\begin{prop}
A subset constraint is safe.
\end{prop}
\begin{proof}
Let $c_1=(p_1,\ldots,p_n), c_2=(q_1,\ldots,q_n) \in \mathcal{S}$, let $c_2 \preceq c_1$, and let $c_2 \in D_{i,X}$.
We have that $q_i \preceq p_i$ for all $1 \le i \le n$.
We assumed that $c_2 \in D_{i,X}$ and thus $q_i \in X$.
Because of the assumption on $Q_i$, we thus have that $p_i = q_i$, and, therefore, that $p_i \in X$.
Thus, $(p_1,\ldots,p_n) \in D_{i,X}$.
\end{proof}

\begin{exa}[Horizontal aggregation]\label{ex:comp2}
In this example, we aggregate a Transport with a Scaler.
The key difference with the previous example is that we have to match the output of the Transport with the input of the Scaler. We call this a horizontal aggregation.
In this example, we consider a Transport with two configurations: one for 60Hz FHD input, and one for 60Hz HD input.
The required budget refers to the needed connection bandwidth in~Gb/s.
\[
\begin{array}{rccccccl}
t_1 = (& \mathrm{HD@60}, & \mathrm{HD@60}, & 2, & \bot, & \bot &  \mathrm{HD@60} &)\\
t_2 = (& \mathrm{FHD@60}, & \mathrm{FHD@60}, & 4, & \bot, & \bot &  \mathrm{FHD@60} &)\\
\end{array}
\]
For the Scaler, we look at the following configurations:
\[
\begin{array}{rccccccl}
s_1 = \big(& \mathrm{FHD@60}, & \mathrm{{HD+}@60}, & (201,15), & \bot, & 60, & \mathrm{HD+} & \big) \\
s_2 = \big(& \mathrm{FHD@60}, & \mathrm{{HD}@30}, & (145,15), & \bot, & 30, & \mathrm{HD} &\big) \\
s_3 = \big(& \mathrm{{HD+}@60}, & \mathrm{HD@60}, & (135,13), & \bot, & 60, & \mathrm{HD} &\big) \\
\end{array}
\]

Taking the free component product results in a configuration set with six configurations (the left-hand side of these configurations is the Transport component):
\[
\begin{array}{ll}
a_1 =  (\mathrm{HD@60}, \mathrm{HD@60}, 2, \bot, \bot, \mathrm{HD@60}, & \mathrm{FHD@60}, \mathrm{{HD+}@60}, (201,15), \bot, 60, \mathrm{HD+}) \\
a_2 =  (\mathrm{HD@60}, \mathrm{HD@60}, 2, \bot, \bot, \mathrm{HD@60}, & \mathrm{FHD@60}, \mathrm{{HD}@30}, (145,15), \bot, 30, \mathrm{HD}) \\
a_3 =  (\mathrm{HD@60}, \mathrm{HD@60}, 2, \bot, \bot, \mathrm{HD@60}, & \mathrm{{HD+}@60}, \mathrm{HD@60}, (135,13), \bot, 60, \mathrm{HD}) \\
a_4 =  (\mathrm{FHD@60}, \mathrm{FHD@60}, 4, \bot, \bot, \mathrm{FHD@60}, & \mathrm{FHD@60}, \mathrm{{HD+}@60}, (201,15), \bot, 60, \mathrm{HD+}) \\
a_5 =  (\mathrm{FHD@60}, \mathrm{FHD@60}, 4, \bot, \bot, \mathrm{FHD@60}, & \mathrm{FHD@60}, \mathrm{{HD}@30}, (145,15), \bot, 30, \mathrm{HD}) \\
a_6 =  (\mathrm{FHD@60}, \mathrm{FHD@60}, 4, \bot, \bot, \mathrm{FHD@60}, & \mathrm{{HD+}@60}, \mathrm{HD@60}, (135,13), \bot, 60, \mathrm{HD}) \\
\end{array}
\]

The stream that needs to be processed is of type FHD@60, and we use a subset constraint on the Transport parameter, $D_{6, \{ \mathrm{FHD@60} \}}$, to select only the relevant transport configurations, which are $a_4$, $a_5$ and $a_6$.
We also need to match the output of the Transport to the input of the Scaler.
Since the output type of the Transport is equal to the input type of the Scaler, we use the identity function $\mathbf{id}$ (which clearly is a $\succeq$-derivation) to create a producer-consumer constraint $D_{2,7,\mathbf{id}}$.
Then we see that $a_6 \not\in D_{2,7,\mathbf{id}}$, since output FHD@60 does not match input HD+@60 (with $=$ as the ordering).
This removes $a_6$.
Then we also apply a parameter constraint, $D_{12, \{ \mathrm{HD} \}}$, to select only those configurations with an HD output. That removes $a_4$ since it has HD+ output.

The derivations for the parts are the following.
We use $\mathit{copy}_1$ for the input, since the input of the aggregation is the input of the Transport.
We use $\mathit{copy}_8$ for the output, since the output of the aggregation is the output of the Scaler.
This models that the output of the transport is  consumed completely by the Scaler.
We use $\mathit{group}_{3, 9}$ for the required budget (collecting the required budgets of the two constituent interfaces), $\mathit{void}$ for the provided budget,
$\mathit{copy}_{11}$ for the quality (indicating that the aggregation inherits its quality from the Scaler) and $\mathit{void}$ for the parameters.
This last voiding models abstraction of the Transport and Scaler parameters as they have served their purpose and are not relevant anymore.
Alternatively, one could decide that the aggregation copies the parameters from the Scaler, to allow configuration of the output of streams as a whole.
Finally, we abstract from the first twelve dimensions which contain the individual component parts and minimize.
This gives the following configuration:
\[
a_5 =  (\mathrm{FHD@60}, \mathrm{HD@30}, (4, (145,15)), \bot, 30, \bot)
\]
This aggregation formally is the following instance of Eq.~\ref{eq:aggr}:
\[
	\begin{array}{l}
\mathit{min} \circ (\downarrow 1)^{12} \circ \mathit{void} \circ \mathit{copy}_{11} \circ \mathit{void} \circ \mathit{group}_{3,9} \circ \mathit{copy}_8 \circ \mathit{copy}_1 \circ{} \\
\hspace*{5cm} \cap D_{12, \{ \mathrm{HD} \}} \circ \cap D_{6, \{ \mathrm{FHD@60} \}} \circ \cap D_{2,7,\mathbf{id}} \circ \times
	\end{array}
\]
\end{exa}

This example shows how we can use parameters to select subsets of configurations, and how to use a producer-consumer constraint for i/o matching (budget matching works in a similar way).
The output in the example is completely consumed. This, however, needs not be the case in general. Modeling broadcast, for instance, leaves the broadcasting output untouched in an aggregation.
In general, we need non-trivial $\preceq$-derivations for the aggregation of the parts, for instance, to calculate the remaining budget when a required budget is strictly smaller than a provided budget. The remaining budget is then provided to other components by the aggregation.
Another example is the aggregation of two image-processing components that both have latency as quality. For the aggregation, we want to add the latencies of the individual components.
Real numbers can often be used to express inputs, outputs, budgets and qualities. The kind of arithmetic operations that we need for the proper $\preceq$-derivations are supported by the following definitions.

\begin{defi}[Addition-Multiplication-Min-Max]\label{def:add}
Let $\mathcal{S} = Q_1 \times \ldots \times Q_n$ be a configuration space, let $1 \le i,j \le n$, let $Q_i = \mathbb{R}$ and $Q_j = \mathbb{R}$, and let $\preceq_{Q_i} \, =\, \le$ and $\preceq_{Q_j} \,=\, \le$.
We define $\mathit{add}_{i,j} : \mathcal{S} \to Q$ with  $Q = \mathbb{R}$ and $\preceq_Q \,=\, \le$ as $\mathit{add}_{i,j}(q_1, \ldots, q_n) = q_i + q_j$.
We define $\mathit{mult}$, $\mathit{min}$ and $\mathit{max}$ in a similar way.
\end{defi}
\begin{prop}
Addition/multiplication/min/max are $\preceq$-derivations.
\end{prop}
\begin{proof}
Let $(q_1, \ldots, q_n),  (q_1', \ldots, q_n')\in \mathcal{S}$ and let $(q_1, \ldots, q_n) \preceq (q_1', \ldots, q_n')$.
We have to show that $f(q_i, q_j) \le f(q_i', q_j')$ with $f$ being addition, multiplication, min or max.
This follows straightforwardly from the assumption that $q_i \le q_i'$ and $q_j \le q_j'$.
\end{proof}

For subtraction and division, the situation is subtly different as the ordering on one of the posets needs to be reversed.
Subtraction can be used, for instance, to compute the remaining provided budget when a budget provider is aggregated with a budget consumer. This, however, nicely matches with the alternating interpretation of input and output, and provided budget and required budget as we will see in the example below.
In the following definition, $Q_p$ is typically the producer (output or provided budget), and $Q_c$ is the consumer (input or required budget).

\begin{defi}[Subtraction-Division]\label{def:subdiv}
Let $\mathcal{S} = Q_1 \times \ldots \times Q_n$ be a configuration space, let $1 \le p \neq c \le n$, let $Q_p = Q_c = \mathbb{R}$, and let $\preceq_{Q_p} \,=\, \le$ and $\preceq_{Q_c} \,=\, \ge$.
We define $\mathit{sub}_{p,c} : \mathcal{S} \to Q$ with  $Q = \mathbb{R}$ and $\preceq_Q \,=\, \le$ as $\mathit{sub}_{p,c}(q_1, \ldots, q_n) = q_p - q_c$.
We define $\mathit{div}$ in a similar way, but only for $Q_p = Q_c = \mathbb{R}^{\ge 0}$.
\end{defi}
\begin{prop}
Subtraction/division are $\preceq$-derivations.
\end{prop}
\begin{proof}
Let $(q_1, \ldots, q_n),  (q_1', \ldots, q_n')\in \mathcal{S}$ and let $(q_1, \ldots, q_n) \preceq (q_1', \ldots, q_n')$.
We have to show that $f(q_p, q_c) \le f(q_p', q_c')$ with $f$ being subtraction or division.
We have assumed that $q_p \preceq q_p'$, which means that $q_p \le q_p'$.
We also have assumed that $q_c \preceq q_c'$, which means that $q_c \ge q_c'$.
Clearly, $q_p - q_c \le q_p' - q_c'$, and (assuming non-negative numbers) $q_p / q_c \le q_p' / q_c'$.
\end{proof}

\begin{exa}[Vertical aggregation]\label{ex:comp3}
In this example, we aggregate a Fiber component that provides 10~Gb/s \emph{bandwidth} and a Connection component that requires $4$~Gb/s \emph{bandwidth} and provides $4$~Gb/s \emph{connection bandwidth}.
We connect the provided budget of the Fiber with the required budget of the Connection and therefore call this a vertical aggregation.
\[
\{ f_1 \} = \{ ( \bot, \bot, \bot, 10, \bot, \bot) \}
\]
\[
\{ c_1 \} = \{ ( \bot, \bot, 4, 4, \bot, \bot) \}
\]
\[
\{f_1\} \times \{c_1\} = \{ ( \bot, \bot, \bot, 10, \bot, \bot, \bot, \bot, 4, 4, \bot, \bot) \}
\]

Application of the budget constraint $D_{4,9,\mathbf{id}}$ does not filter out this configuration because $4 \le 10$.
Then, we apply $\mathit{void}$ three times to model the input and output derivation, and the derivation for the required budget.
We use $\mathit{sub}_{4,9}$ (on the natural numbers) to model the consumption of the budget and compute the remaining provided \emph{bandwidth} budget.
Note that the requirements that $Q_4 = Q_9 = \mathbb{N}$, $\preceq_{Q_4} \,=\,  \le$ and $\preceq_{Q_9} \,=\,  \ge$ are met because of the alternating approach to modeling provided and required budgets.
We group this remaining provided \emph{bandwidth} with the provided \emph{connection bandwidth} and remove the superfluous dimension.
Finally, we apply $\mathit{void}$ twice to model the composition of the quality and parameter parts.
As last steps, we abstract from the first twelve dimensions and minimize.
This gives the following configuration, capturing that the aggregate component provides a combination of 4 Gb/s \emph{connection bandwidth} and 6 Gb/s raw \emph{bandwidth}:
\[
( \bot, \bot, \bot, (4,6), \bot, \bot)
\]
This aggregation formally is the following instance of Eq.~\ref{eq:aggr}:

\[
\mathit{min} \circ (\downarrow 1)^{12} \circ\mathit{void}^2 \circ (  \downarrow 16 \circ \mathit{group}_{10,16} \circ \mathit{sub}_{4,9} ) \circ \mathit{void}^3 \circ \cap D_{4,9,\mathbf{id}} \circ \times
\]
\end{exa}

So far, we have given a number of operations that we can use for aggregation.
The free product combines components and creates all possible combinations of the individual configurations.
Abstraction can be used to remove dimensions that are not relevant anymore.
Function application of $\preceq$-derivations defined above gives us operations that can be used to model the aggregation of input, output, etcetera. Finally, the producer-consumer constraint can be used to match input with output or required budget with provided budget, and the subset constraint can be used to select an arbitrary subset of configurations based on an unordered dimension (e.g., the parameter dimension).
As shown in Ex.~\ref{ex:comp1}, Ex.~\ref{ex:comp2} and Ex.~\ref{ex:comp3}, aggregation functions can be constructed from these basic building blocks using the template of Eq.~\ref{eq:aggr}.

\subsection{Aggregation templates}\label{sec:aggr-patterns}
The examples in the previous section have shown three templates that can be used for aggregation: free aggregation, horizontal aggregation and vertical aggregation.
These templates can be realized by using the general aggregation pattern of Eq.~\ref{eq:aggr} with proper derivations to implement $f_i$ etc., and with constraints for i/o and budget matching, parameter selection, and domain-specific filtering using the quality dimension.
Next, we show how we can realize these templates in terms of generalized addition and subtraction operators.

Let us assume that we know how to add and subtract, i.e., we have functions $+$ and $-$ with signature $Q \times Q \to Q$ with the following properties for every poset $Q$:
\[
a \preceq_{Q} c \land b \preceq_{Q} d  \Rightarrow  a +_Q b \preceq_{Q} c +_Q d
\]
\[
a \preceq_{Q} c \land b \succeq_{Q} d  \Rightarrow  a -_Q b \preceq_{Q} c -_Q d
\]
These properties hold for the regular addition and subtraction of numbers.
These operations give rise to $\preceq$-derivations.
Poset addition, defined below, takes care of adding values from identical posets, and also of values from different posets by grouping them together in a tuple.

\begin{defi}[Poset addition]
Let $\mathcal{S} = Q_1 \times \ldots \times Q_n$ be a configuration space, and let $1 \le i,j \le n$.
We define $+_{i,j} : \mathcal{S} \to Q$ as follows:
 $+_{i,j} (q_1, \ldots, q_n) = q_i +_{Q_i} q_j$ with $\preceq_Q = \preceq_{Q_i}$ if $Q_i = Q_j$ and $\preceq_{Q_i} = \preceq_{Q_j}$
 and otherwise $+_{i,j} (q_1, \ldots, q_n)=(q_i, q_j)$ with $(q_i, q_j) \preceq_Q (q_i',q_j')$ if and only if  $q_i \preceq_{Q_i} q_i'$ and $q_j \preceq_{Q_j} q_j'$.
\end{defi}
\begin{prop}
Poset addition is an $\preceq$-derivation.
\end{prop}
\begin{proof}
Let $(q_1, \ldots, q_n), (q_1',\ldots,q_n') \in \mathcal{S}$ and let $(q_1, \ldots, q_n) \preceq (q_1',\ldots,q_n') $.
First, assume that $Q_i = Q_j$ and $\preceq_{Q_i} = \preceq_{Q_j}$.
Then, $+_{i,j}(q_1, \ldots, q_n) = q_i +_{Q_i} q_j$ and $+_{i,j}(q_1', \ldots, q_n') = q_i' +_{Q_i} q_j'$.
From the assumption that $q_i \preceq_{Q_i} q_i'$ and $q_j \preceq_{Q_i} q_j'$, it follows with the definition of $+_{Q_i}$ that $q_i +_{Q_i} q_j \preceq_{Q_i} q_i' +_{Q_i} q_j'$.
Second,
assume that  $Q_i \neq Q_j$ or $\preceq_{Q_i} \neq \preceq_{Q_j}$.
Then $+_{i,j}(q_1, \ldots, q_n) = (q_i, q_j)$ and $+_{i,j}(q_1', \ldots, q_n') = (q_i',q_j')$.
Clearly, $(q_i, q_j) \preceq_Q (q_i',q_j')$ since $q_i \preceq_{Q_i} q_i'$ and $q_j \preceq_{Q_j} q_j'$.
\end{proof}

The case for subtraction is more complicated due to the alternating interpretation but similar to the situations of the producer-consumer constraint of Def.~\ref{def:pc} and the subtraction-division of Def.~\ref{def:subdiv}.

\begin{defi}[Poset subtraction]
Let $\mathcal{S} = Q_1 \times \ldots \times Q_n$ be a configuration space, let $1 \le i,j \le n$, let $Q_i = Q_j$, and let $\preceq_{Q_i} = \preceq_{Q_j}^{-1}$.
We define $-_{i,j} : \mathcal{S} \to Q$ as follows:
 $-_{i,j} (q_1, \ldots, q_n) = q_i -_{Q_i} q_j$ with $\preceq_Q = \preceq_{Q_i}$.
\end{defi}
\begin{prop}
Poset subtraction is an $\preceq$-derivation.
\end{prop}
\begin{proof}
Let $(q_1, \ldots, q_n), (q_1',\ldots,q_n') \in \mathcal{S}$ and let $(q_1, \ldots, q_n) \preceq (q_1',\ldots,q_n') $.
Then we have that $-_{i,j}(q_1, \ldots, q_n) = q_i -_{Q_i} q_j$ and $-_{i,j}(q_1', \ldots, q_n') = q_i' -_{Q_i} q_j'$.
We must show that $q_i -_{Q_i} q_j \preceq_{Q_i} q_i' -_{Q_i} q_j'$.
From our assumptions, we have that  $q_i \preceq_{Q_i} q_i'$ and $q_j \preceq_{Q_j} q_j'$.
Since $Q_i = Q_j$ and $\preceq_{Q_i} = \preceq_{Q_j}^{-1}$, we have that $q_j \succeq_{Q_i} q_j'$.
With the definition of $-_{Q_i}$, this gives that $q_i -_{Q_i} q_j \preceq_{Q_i} q_i' -_{Q_i} q_j'$.
\end{proof}

With these operations on posets, we can then define free, horizontal and vertical aggregation as follows (where the finishing steps, minimization and abstraction of the twelve dimensions and filtering on parameters or quality constraints, are omitted):
\begin{align*}
\parallel \quad & \definitie &
 +_{6,12} \circ +_{5,11} \circ +_{4,10} \circ +_{3,9} \circ +_{2,8} \circ +_{1,7} \circ \times \\
\Rightarrow \quad & \definitie &
 +_{6,12} \circ +_{5,11} \circ +_{4,10} \circ +_{3,9} \circ ( \downarrow 14 \circ +_{8,14} \circ -_{2,7}) \circ \mathit{copy}_1 \circ \cap D_{2,7,\mathbf{id}} \circ \times \\
\Uparrow \quad & \definitie &
 +_{6,12} \circ +_{5,11} \circ  ( \downarrow 16 \circ +_{10,16} \circ -_{4,9}) \circ \mathit{copy}_3 \circ +_{2,8} \circ +_{1,7} \circ \cap D_{4,9,\mathbf{id}} \circ \times
\end{align*}

Horizontal aggregation is only defined if the output type of the first component is equal to the input type of the second component, and if the orderings on input and output are reversed (due to the requirement of the poset subtraction).
Our definition is such that the input of the second component is completely covered by the output of the first component (the producer-consumer constraint). What remains of the output of the first component is added to the output of the second component resulting in the output of the aggregation.
The definition of vertical aggregation is similar, only the roles of input and output are played by required and provided budget, respectively.
The three aggregation templates are shown schematically in Fig.~\ref{fig:free}, Fig.~\ref{fig:horizontal} and Fig.~\ref{fig:vertical}, respectively.
Note that these are specific example templates of how we can define aggregation with the addition and subtraction operations on posets. For instance, the specific definitions for the consumption of output and provided budget is a choice. We can imagine that there are situations where variations of these aggregation templates are needed. These can then be defined with the flexible mathematical tools that we have presented.

\begin{figure}
\centering
\includegraphics[width=0.7\linewidth]{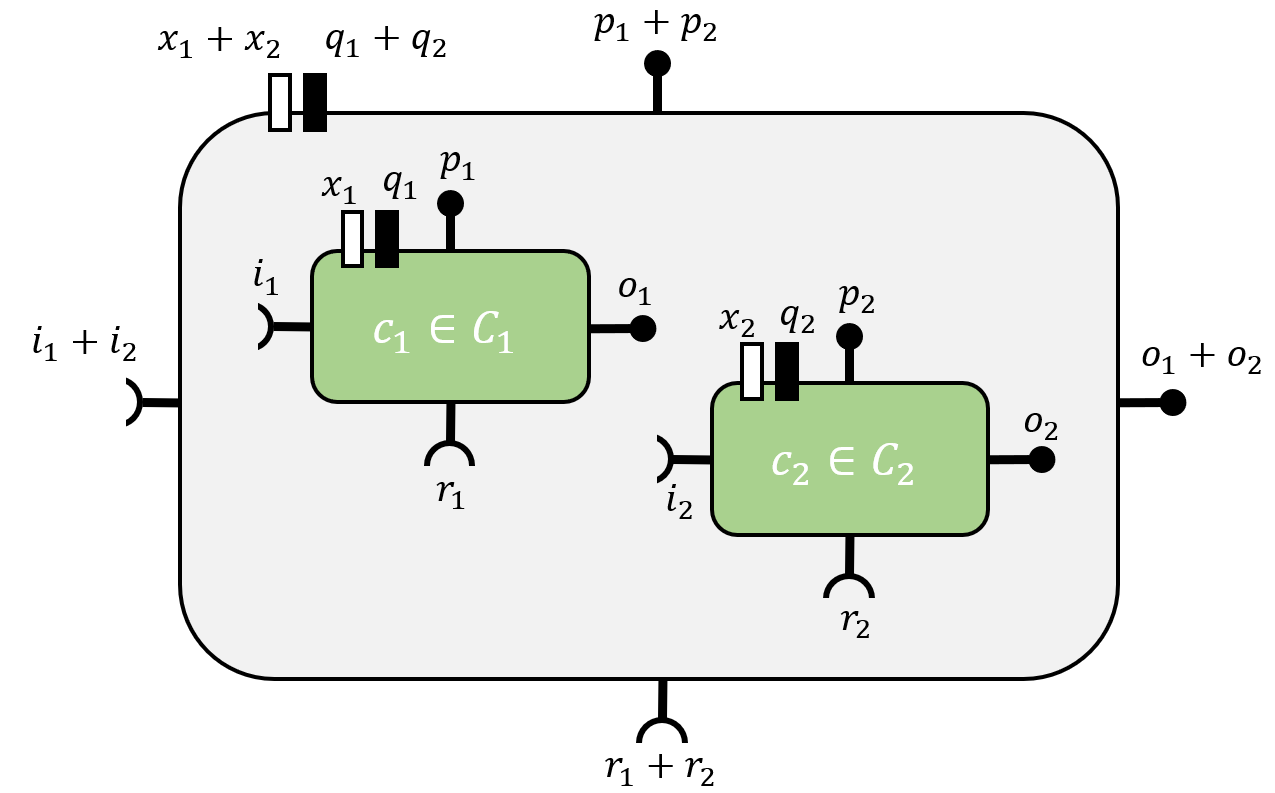}
\caption{Abstract representation of free aggregation.}\label{fig:free}
\end{figure}

\begin{figure}
\centering
\includegraphics[width=0.7\linewidth]{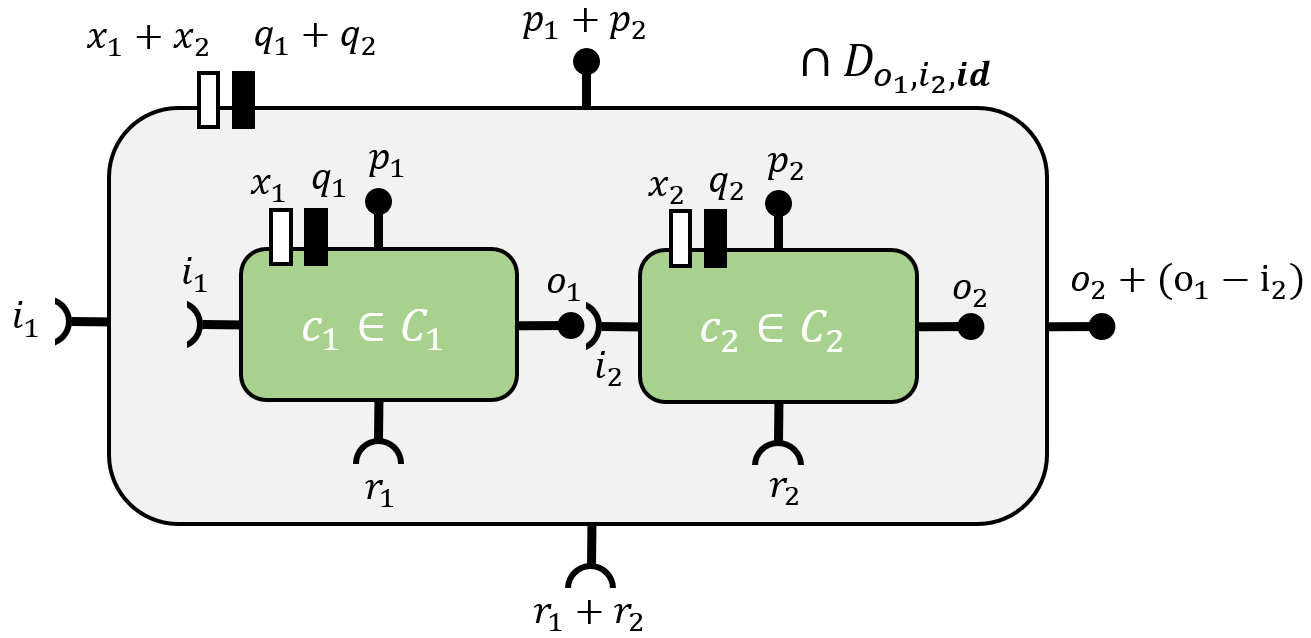}
\caption{Abstract representation of horizontal aggregation.}\label{fig:horizontal}
\end{figure}

\begin{figure}
\centering
\includegraphics[width=0.7\linewidth]{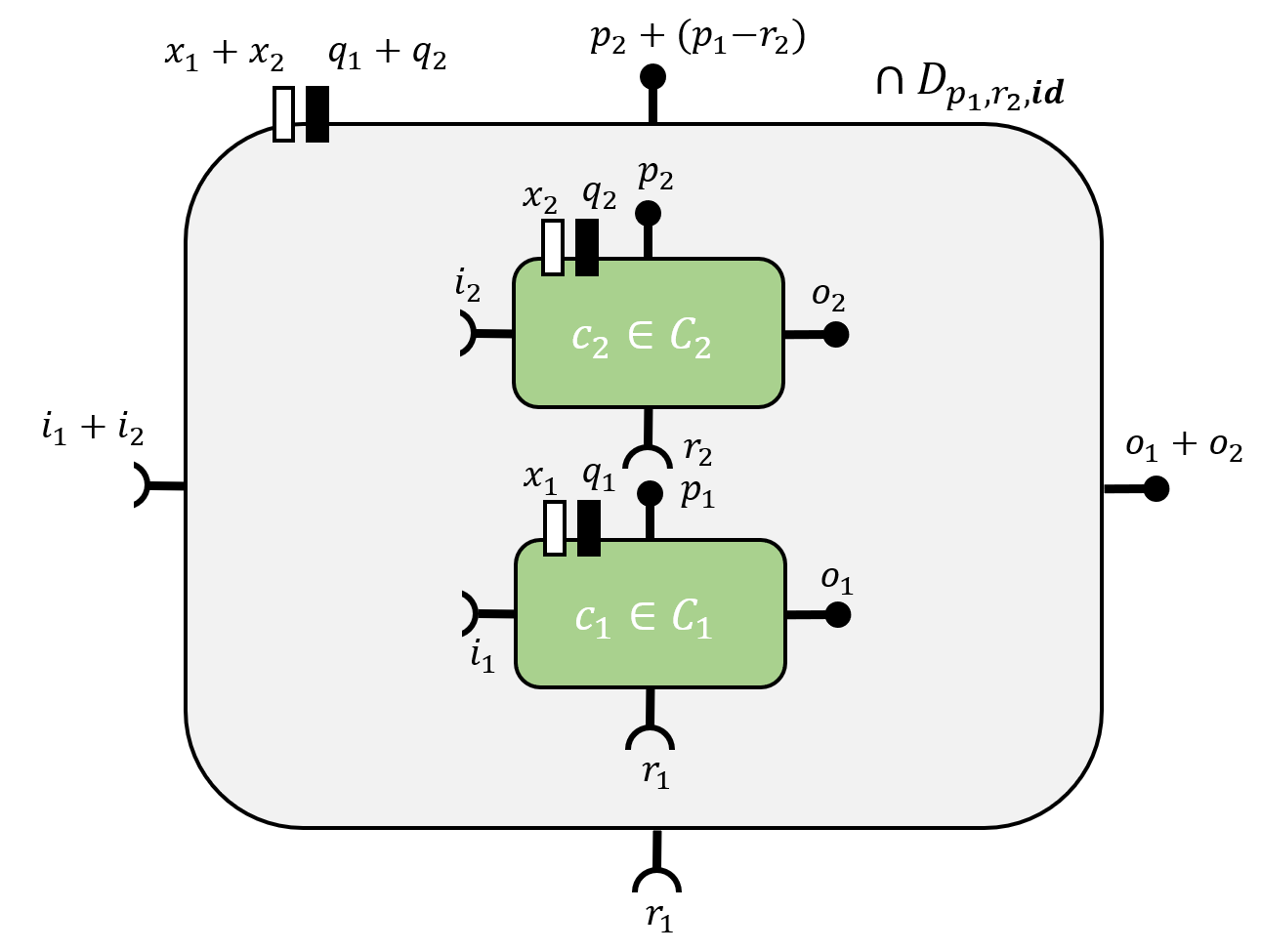}
\caption{Abstract representation of vertical aggregation.}\label{fig:vertical}
\end{figure}

When not everything is known at design time, which is the case in our example because the video streams that need processing are only known at runtime, an online quality and resource manager is needed.
In the following section we develop such a manager which is fully tailored to the video-processing example and implemented in the Java programming language.


\section{Quality and resource management for the video-processing system}\label{sec:qrm}

In this section, we put all theory into practice and create a Java solution for the quality- and resource-management problem as described in Ex.~\ref{ex:1}.
We emphasize that this is not a general approach to quality and resource management, but an example of how the mathematical tools that we have provided can be used in practice to solve a specific problem.
The first step was to implement the posets that we need, e.g., \emph{video}, \emph{bandwidth},  etcetera (see Fig.~\ref{fig:ex1}).
We capture the posets in the \texttt{Value} type (a Java interface), which provides ordering methods to compare the \texttt{Value} with another \texttt{Value}, and also the $+$ and $-$ operations defined in Sec.~\ref{sec:aggr-patterns}.
We have an abstract \texttt{NatValue} implementation of this interface for natural numbers, and several extensions of \texttt{NatValue} to model the concrete types that we need, e.g., \emph{bandwidth}, \emph{hres}, \emph{vres}, etcetera.
Furthermore, we have a \texttt{Tuple} implementation of the \texttt{Value} interface that realizes tuples of arbitrary length and nesting, with other \texttt{Value} instances at its leafs.
We have several extensions of this \texttt{Tuple} type to model the multi-dimensional base types that we need, e.g., \emph{video} and \emph{scalers}.

As second step, we implemented the Pareto-algebraic machinery specialized to quality and resource management. We  defined an interface \texttt{ConfigurationSet} that components can implement. It has a single method to get all configurations of the component, which are implementations of the \texttt{Configuration} interface. This type provides the methods to obtain \texttt{Value} instances for input, output, etcetera.
Furthermore, we  implemented all Pareto-algebraic operations that we need, e.g., minimization, abstraction, etcetera. We  also implemented the derived functions that we need, e.g., the free, horizontal and vertical aggregation using the $+$ and $-$ operations on the posets, as described above.

As a third step, we  created an application-specific QRM algorithm to map Transport and Scaler components for $n$ streams on $k$ Fiber and $k$ HW scaler components using appropriate VEP components.
We partitioned the platform resources into execution-platform components. Each is a free aggregation of a Fiber and a HW scaler component.
A stream is modeled by a single application component, which is the horizontal aggregation of a Transport component with a Scaler component, see Ex.~\ref{ex:comp2}.
The configurations of the Scaler are determined as follows:
\begin{itemize}[label=$\triangleright$]
\item The input can be FHD, HD+, HD, qHD (960x540 pixels) or nHD (640x360 pixels) with a rate of 90, 60, 30, 20, 15, 10, 6, or 5~Hz.
\item The output can be as the input, provided that the resolution of the output is smaller than the resolution of the input, and the rate of the output divides the rate of the input.
\item The required \emph{scaling} budget consists of two parts. The \emph{comp} part is calculated as the input pixel rate plus the output pixel rate.
The \emph{segs} part equals the horizontal resolution divided by 128 (rounded up).
\item The quality equals the output frame rate.
\item The parameters equal the video parameters of the output.
\end{itemize}
This results in 270 configurations for each Scaler component.

Every configuration of an application component has well-defined required budgets (determined during system design).
The QRM aggregates from the provided physical platform budgets VEPs that provide budgets matching the required application budgets.
The QRM also computes the bindings of the VEPs to the physical resources (Fiber/HW scaler compositions in this example).

\begin{algorithm}[!th]
\caption{$\mathit{VideoProcessingQRM}(k, (v_1, x_1) \ldots, (v_n, x_n), \mathit{cost})$}\label{alg:opt}
\begin{algorithmic}[1]
\STATE{$R \subseteq \mathbb{N}^n := \emptyset$}
\FOR{$i = 1$ \textbf{to} $k$}
	\STATE{$\mathit{ex}_i := f \parallel hw$}
\ENDFOR
\FOR{$i = 1$ \textbf{to} $n$}
	\STATE{$\mathit{app}_i :=( t \Rightarrow sc) \cap D_{6,\{(v_i,x_i)\}} $}
	\STATE{create Connection $c_i$}
	\STATE{create Virtual scaler $\mathit{vsc}_i$}
	\STATE{$\mathit{vep}_i := c_i \parallel \mathit{vsc}_i$}
	\STATE{$\mathit{va}_i := \mathit{vep}_i \Uparrow \mathit{app}_i$}
\ENDFOR
\FORALL{mappings $m : \{ 1, \ldots, n\} \to \{ 1, \ldots, k \}$}
	\FOR{$i = 1$ \textbf{to} $n$}
		\STATE{$\mathit{ex}_{m(i)} := \mathit{ex}_{m(i)} \Uparrow \mathit{va}_i $}
	\ENDFOR
	\STATE{$S := \pi \circ \mathit{min} ( \downarrow 2 \circ (\downarrow 1)^4 \circ (\mathit{ex}_1 \parallel \ldots \parallel \mathit{ex}_k))$}
	\STATE{$R := \mathit{min} ( R \cup S)$}
\ENDFOR
\RETURN{select one configuration from $\mathit{min}(\mathit{cost}(R))$}
\end{algorithmic}
\end{algorithm}

Algorithm~\ref{alg:opt} shows the QRM algorithm for the video-processing system.
This algorithm explores all configurations of the aggregation of stream and platform components, as illustrated in Fig. ~\ref{fig:ex1}, and returns a cost-optimal configuration as a solution to the given multi-objective optimization problem.
The algorithm  is assumed to have access to all the interface models of the application and platform components (Transport ($t$), Scaler ($sc$), Fiber ($f$), HW Scaler ($hw$)).
It has as parameters the number of execution-platform components $k$, $n$ tuples $(v_i, x_i)$, where $v_i$ is a \emph{video} value that specifies the resolution and rate of the incoming stream $i$, and $x_i$ is a \emph{video params} value that specifies the required output resolution after scaling of stream $i$, and a $\mathit{cost}$ function.
The cost function is an $\preceq$-derivation on the configuration space $\mathbb{N}^n$ that captures the output rate of each stream.
The Pareto-minimal solutions are captured in set $R$, which is initialized on line 1.
It contains the quality of each stream (the output rate in Hz).
Lines 2--4 initialize the $k$ execution-platform components.
Lines 5--11 initialize the $n$ application components and their VEPs.
In line 6, we apply a subset constraint to select the proper parameters for the stream.
The proper VEP for the application component is computed in lines 7--9, and in line 10 we create the vertical aggregation of the application and the VEP component.
The for-loop in lines 12--18 iterates over all possible mappings of the $n$ application components on the $k$ execution-platform components.
An entry $m$ gives for an application component $i$ the execution-platform component $m(i)$ to which it is {bound} via a VEP.
Lines 13--15 then {bind} the application-VEP combinations to the execution-platform components using vertical aggregation.
In each iteration of this loop, one application-VEP combination is bound to the (remaining) budget of the platform component on which it is mapped.
After this loop, $\mathit{ex}_i$ is an aggregation of an execution-platform component with zero or more application-VEP components.
Line 16  combines all these $\mathit{ex}_i$ aggregations through free aggregation.
The abstractions filter out everything but the qualities of the streams.
A permutation $\pi$ is applied such that the quality of stream $i$ is on the $i$-th position in the resulting configuration space.
These configurations then are combined with the alternatives of other {mappings} in line 17.
Finally, line 19 applies the cost function and selects one of the optimal solutions.

In principle, we enumerate all possible mappings of streams onto hardware components. Of course, this can be a huge number, and therefore we have applied \emph{symmetry reduction}~\cite{symred}.
Since the execution-platform components are totally interchangeable, we can define an equivalence relation on the {mappings} $m$ in line 12, and we only need to explore one aggregation per equivalence class.
For instance, if we have two hardware components, say $\mathit{ex}_1$ and $\mathit{ex}_2$,
and two streams as in line 10 of the algorithm, say $\mathit{va}_1$ and $\mathit{va}_2$,
then the mapping $\mathit{ex}_1 \Uparrow \mathit{va}_1$ and $\mathit{ex}_2 \Uparrow \mathit{va}_2$
gives the same result as the mapping $\mathit{ex}_2 \Uparrow \mathit{va}_1$ and $\mathit{ex}_1 \Uparrow \mathit{va}_2$
since $\mathit{ex}_1$ and $\mathit{ex}_2$ are interchangeable.
We can also apply symmetry reduction to the streams by assuming that streams of the same type (determined by their input and required output) are completely interchangeable.
Of course, the soundness of this second reduction depends on the cost function that is ultimately used.
For instance, one stream may be more important for the user than another stream with equal parameters.
This difference in priority is then typically expressed in the cost function, and clearly such streams are not interchangeable.
Note that this application of symmetry reduction is an ad-hoc optimization specific for this case and the example QRM as shown in Alg.~\ref{alg:opt}.
A systematic approach to symmetry reduction is, for now, out of scope of our mathematical framework.
In Sec.~\ref{sec:qrml}, we introduce a domain-specific language for the specification of the QRM interfaces of components and
their composition. We can imagine that further developments of this language may support syntactic constructs that indicate
symmetries in the search space.

Now consider the situation where $n=6$, three HD+@90 to HD streams and three HD@60 to qHD streams, and $k=3$.
We use the cost function that maximizes the minimal rate.
First, we observe the effect of symmetry reduction. Without symmetry reduction, we have 690 different {mappings}. Using symmetry reduction for the three execution-platform components results in 115 possibilities. Also using symmetry reduction for streams of the same \emph{video} type yields only 18 possibilities.
With or without stream symmetry reduction, we obtain a single configuration with the following quality at the end of the algorithm:
\[
( 90, 90, 90, 60, 60, 60 )
\]

This thus shows that we can process all six streams without rate loss.
Each execution-platform component {hosts} an HD+@90 stream and an HD@60 stream.
Now assume that we need to additionally scale a qHD@60 stream to nHD.
The number of {mappings} with only symmetry reduction for the execution-platform components is 315, and the number of possibilities when also stream symmetry reduction is used, is 40.
Optimization with both forms of symmetry reduction results in three Pareto-minimal configurations:
\begin{align*}
c_1 = (30 , 90 , 90 , 60 , 60 , 60 , 60 ) \\
c_2 = (90 , 90 , 90 , 20 , 60 , 60 , 20 ) \\
c_3 = (90 , 90 , 90 , 15 , 60 , 60 , 30 )
\end{align*}
All three configurations come from the same {mapping}: two execution-platform components process an HD+@90 and an HD@60 stream each, and the third additionally processes the qHD@60 stream.
The cost function now selects $c_1$ since it is the unique optimum configuration, i.e., it has a higher minimum output rate of 30 Hz among the streams  than the other two configurations.
Optimization without symmetry reduction for streams yields nine Pareto-minimal configurations at the end of line 18.
Each of the three configurations above  comes in three variations that correspond to permutations of the streams of the same type.

With seven streams in the system, and 270 configurations for each Scaler, the raw number of possible configuration combinations equals $270^7 = 104,603,532,030,000,000$.
The total number of combinations of $\mathit{app}_i$ configurations (line 6 of the algorithm; determined by the possible output rates for each stream) equals $6^7 = 279,936$.
Clearly, a naive approach to optimization can end badly, that is, not end at all for all practical purposes.
Using the mathematical framework from this paper, our Java QRM implementation optimizes the seven streams using symmetry reduction in approximately $200$ ms on a standard laptop computer (40 mapping possibilities). Without stream symmetry reduction, the average execution time is around $800$ ms (315 possibilities).
We believe that the implementation can still be improved significantly to reduce the running time. Our conclusion therefore is that our framework gives apt mathematical tools for runtime quality and resource management.


\section{The QRML specification language}\label{sec:qrml}
In \cite{qrml}, we have introduced the Quality- and Resource-Management Language (QRML) and toolset.
In this section, we illustrate (an improved version of) the component language for specification of QRM interfaces and their composition
(online tooling available at \url{https://qrml.org}).
Three fundamental parts of the mathematical framework presented in this paper are (i) partially-ordered sets to define the configuration spaces of the components, (ii) composition as an aggregation of alternatives, and (iii) optimization. In the next subsections, we explain how these are supported by the QRML toolset.

\subsection{Partially-ordered sets in QRML}
The language provides primitive sets of values (integers and reals) and enumerations, and supports the construction of combinations of these sets.
Each set of values has a default partial order and the language provides means to manually specify a partial order (via lambda expressions).
A QRML type is a set of values and its partial order, and these types  are used to specify the configuration spaces of components.
By convention, every type in QRML has a least and greatest element, denoted in the syntax by \texttt{bot} and \texttt{top} respectively.

\begin{figure}[!h]
\centering
\includegraphics[width=.9\linewidth]{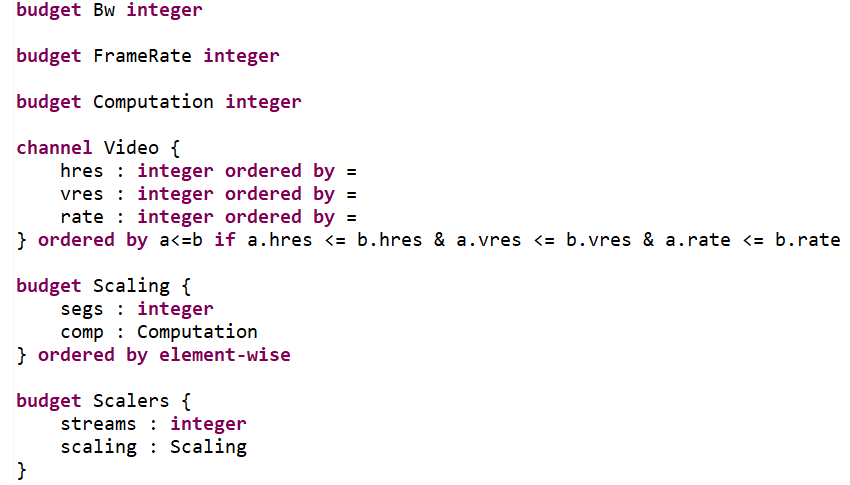}
\caption{Type definitions for the video-processing system in QRML.}\label{fig:qrml-types}
\end{figure}

Some of the types for the video-processing example from Sec.~\ref{sec:ex} are shown in Fig.~\ref{fig:qrml-types}.
A type definition starts with
one of the keywords \texttt{typedef}, \texttt{channel} or \texttt{budget}.
The \texttt{Bw}, \texttt{FrameRate} and \texttt{Computation} budgets are specified as a partial order consisting of the integers and the (default) $\le$ relation.
The \texttt{Video} channel type is a combination of three integers for the horizontal resolution, vertical resolution and frame rate, respectively.
Note that the partial-order relations for each of these three parts is the equality relation on integers.
This implies that, in this model, the values must match exactly and a higher provided resolution is not accepted to match a lower required resolution.
The \texttt{ordered~by} clause gives a predicate that specifies how two elements of the \texttt{Video} type are ordered by the partial-order relation.
This definition uses the partial-order relations of the parts.
Thus, $(h_1,v_1,r_1) \le (h_2,v_2,r_2)$ if and only if $h_1=h_2$ and $v_1 = v_2$ and $r_1 = r_2$.
The \texttt{Scaling} budget also is a combination in which the \texttt{segs} part is ordered by the $\le$ relation on the integers, and the ordering of the \texttt{comp} part is given by the order on the \texttt{Computation} budget.
The \texttt{element-wise} ordering is short-hand for the relation defined similarly as for the \texttt{Video} type (which can alternatively be defined with the \texttt{element-wise} clause).
Finally, the \texttt{Scalers} budget shows another level of nesting of types.
Because the order is not explicitly defined for the \texttt{Scalers} budget, its partial-order relation is the default element-wise comparison using the partial-order relations of the parts.
Note that the inclusion of default orders in the semantics of QRML implies that the \texttt{ordered~by} clauses for the \texttt{Video} and \texttt{Scaling} types could have been omitted. They are given here for illustration purposes.

\subsection{Components and their composition in QRML}
The types can be used to specify QRML components. A component declaration consists of three elements: port declarations, subcomponent declarations, and constraint declarations.
Components have zero or more {ports}  for each QRM part, and each port has a type.
A port is declared by one of the keywords \texttt{provides}, \texttt{requires}, \texttt{input}, \texttt{output}, \texttt{quality}, or \texttt{parameter}.
Multiple ports for the same part are combined in a multi-dimensional poset for that part using the free product.
QRM parts with no explicit port specification default to the void poset.
The port declarations of a component thus specify its configuration space.
Components can have zero or more subcomponents.
Subcomponents are declared by the \texttt{contains} keyword, and refer to components declared elsewhere.
A \texttt{contains} statement can also offer alternatives, e.g., the component contains subcomponent $A$ or subcomponent $B$.
Finally, components can specify a number of {constraints} on the ports of the component and its subcomponents. These constraints are expressions on the port variables, and are specified between curly brackets after port declarations, via the \texttt{from} keyword after port declarations, or via the \texttt{constraint} keyword.

\begin{figure}
\centering
\includegraphics[width=\linewidth]{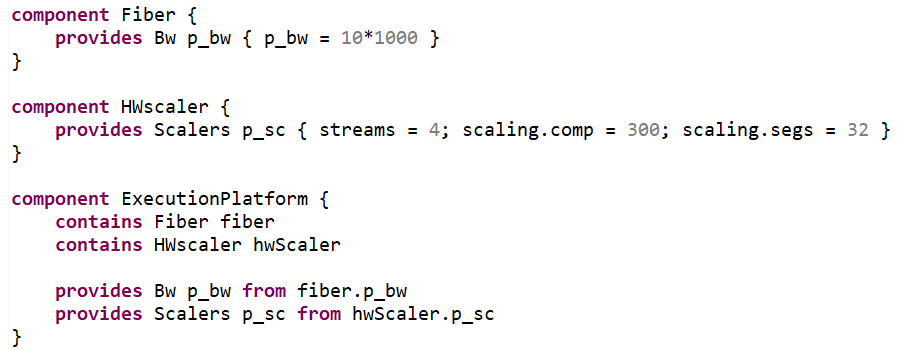}
\caption{Aggregation: the execution platform.}\label{fig:qrml-components}
\end{figure}

For instance, Fig.~\ref{fig:qrml-components} shows how the execution platform of the example can be specified in QRML with the aggregation pattern.
The \texttt{Fiber} and \texttt{HWscaler} components are atomic: they are not built from other components.
They each have a single port and constraints on the values of their port are specified between curly brackets.
The \texttt{ExecutionPlatform} component  has two subcomponents declared by the \texttt{contains} keyword: a  \texttt{Fiber} and a \texttt{HWscaler}.
Furthermore, it declares two ports for the provided part of the QRM interface.
The configuration spaces of the three components are the following:
\[
\mathcal{S}_f = V \times V \times V \times  (\mathbb{Z}\cup \{\bot, \top \}) \times V \times V
\]
\[
\mathcal{S}_h = V \times V \times V \times (\mathbb{Z}\cup \{\bot, \top \})^3 \times V \times V
\]
\[
\mathcal{S}_e = V \times V \times V \times  (\mathbb{Z}\cup \{\bot, \top \})^4 \times V \times V
\]
(where $V = \{ \bot \}$ is the void poset that uses the $=$ partial order).
Both the \texttt{Fiber} and the \texttt{HWscaler} have a single configuration which is specified by their respective port constraints. Their configuration sets are the following:
\[
\begin{array}{l}
C_f = \{ (\bot, \bot, \bot, 10000, \bot, \bot) \} \subseteq \mathcal{S}_f\\
C_h = \{ (\bot, \bot, \bot, (4, 300, 32), \bot, \bot) \} \subseteq \mathcal{S}_h
\end{array}
\]

The \texttt{ExecutionPlatform} specifies an aggregation. We use the constraint-based pattern for its semantics.
The two \texttt{from} constraints relate the provided budget of the aggregate to the provided budgets of the constituents. The aggregation formally is the following instance of Eq.~\ref{eq:aggr2}:
\[
\mathit{min} \circ
(\downarrow 1)^{12} \circ
\cap \{ (x_1, \dots, x_{18}) \,|\, x_{16} = (x_4, x_{10}) \} ( C_f  \times C_h \times \mathcal{S}_e )
\]
Here, $x_{16}$ is the provided budget of the  \texttt{ExecutionPlatform}, $x_4$ is the provided budget of the  \texttt{Fiber}, and $x_{10}$ is the provided budget of the  \texttt{HWscaler}.
In this case, we thus obtain a single configuration for the \texttt{ExecutionPlatform} component with only a composite provided budget:
\[
C_e = \{ (\bot, \bot, \bot, (10000, 4, 300, 32), \bot, \bot) \} \subseteq \mathcal{S}_e
\]

\begin{figure}[htb]
\centering
\includegraphics[width=\linewidth]{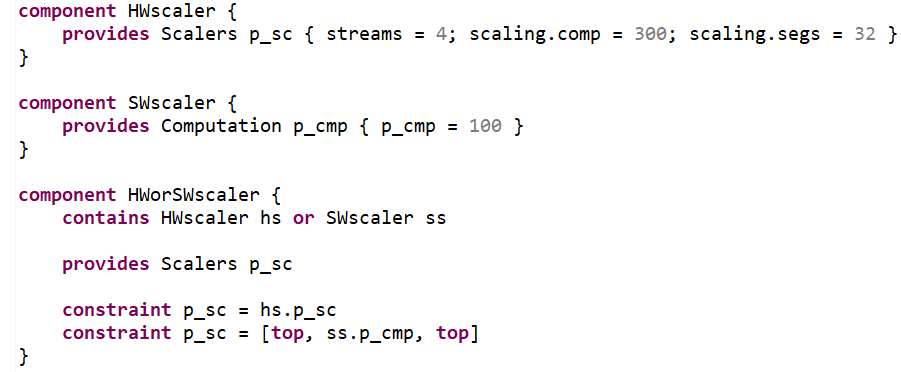}
\vspace*{-3ex}
\caption{Alternatives: a choice between components.}\label{fig:qrml-alt}%
\end{figure}

We use Ex.~\ref{ex:alt} to explain the alternatives construction in QRML.
It introduces the \texttt{SWscaler} component that only provides the \texttt{Computation} budget. The \texttt{HWorSWscaler} component models the choice between the \texttt{HWscaler} and \texttt{SWscaler}. Fig.~\ref{fig:qrml-alt} shows the three relevant components.
The configuration spaces for the \texttt{HWscaler}, \texttt{SWscaler} and \texttt{HWorSWscaler} follow from the port declarations and are the following:
\[
\mathcal{S}_h = V \times V \times V \times (\mathbb{Z} \cup \{\bot, \top \} )^3 \times V \times V
\]
\[
\mathcal{S}_s = V \times V \times V \times (\mathbb{Z}\cup \{\bot, \top \}) \times V \times V
\]
\[
\mathcal{S}_{hs} = V \times V \times V \times (\mathbb{Z}\cup \{\bot, \top \})^3 \times V \times V
\]
The \texttt{contains or} construction in the \texttt{HWorSWscaler} means that this component provides a choice between the \texttt{HWscaler} and \texttt{SWscaler}, i.e., it follows the alternatives pattern.
The configuration sets of the  \texttt{HWscaler} and \texttt{SWscaler} are defined as in Ex.~\ref{ex:alt} and are denoted by $C_{h}$ and $C_{s}$, respectively. Both consist of a single configuration.
To compute the configurations of the \texttt{HWorSWscaler}, we take the alternatives of the two constituent components.
The two constraints give the normalization constraints for Eq.~\ref{eq:alt2}. The alternatives operation formally has the following semantics:
\[
\begin{array}{rl}
(1) & \mathit{min} \circ (\downarrow 1)^{6} \circ \cap \{ (x_1,\ldots,x_{12})\,|\, x_{10} = x_4 \} (C_{h} \times \mathcal{S}_{hs})  \\
(2) & \mathit{min} \circ (\downarrow 1)^{6} \circ \cap  \{ (x_1,\ldots,x_{12})\,|\, x_{10} = (\top, x_4, \top) \} \ (C_{s} \times \mathcal{S}_{hs})  \\
(3) & \mathit{min} \circ \cup
\end{array}
\]
Here, $x_{10}$ is the provided budget of the  \texttt{HWorSWscaler}, and $x_4$ is the provided budget of the  \texttt{HWscaler} in step 1 and the provided budget of the  \texttt{SWscaler} in step 2.
This results in two configurations for the \texttt{HWorSWscaler}:
\[
S_{hs} = \{ (\bot,\bot,\bot, (\top,100,\top), \bot, \bot) , (\bot,\bot,\bot, (4, 300, 32), \bot, \bot) \}
\]

\subsection{Optimization}
We have defined the semantics of a QRML model in terms of configuration sets and Pareto-algebraic operations on these sets.
Each component either is an atomic component (without subcomponents), a choice between subcomponents, or an aggregation of subcomponents.
The first option directly gives the configuration set of the component through its configuration space and the constraints.
The latter two options give the configuration set of the component through instances of the alternatives and aggregation patterns of Eq.~\ref{eq:alt2} and Eq.~\ref{eq:aggr2}, respectively.
The Pareto-minimal configuration set of a component is sufficient and necessary for optimization with respect to an arbitrary cost function (i.e., an arbitrary $\preceq$-derivation).
The actual computation of configuration sets of components is delegated to external tools through proper language transformations (this is work in progress).
For instance, if all used constraints preserve dominance, then we can use Pareto-minimal configuration sets during all computations. This is a good match for a tool such as the Pareto calculator \cite{pareto-calc}. The Pareto-algebraic operations can be implemented sufficiently efficiently for on-line use. QRML specifications may then serve to configure a quality and resource manager, as the one introduced in the previous section.
An alternative for design-time optimization, which does not require dominance preservation of operations, is a constraint-based formulation suitable for a solver such as Z3 \cite{z3}. This formulation can then be used in an optimization approach such as presented in \cite{kneepoints} to approximate the Pareto-minimal set of configurations.


\section{Conclusions}\label{sec:conclusions}
We have presented an interface-modeling framework for quality and resource management of component-based systems.
The QRM interface of components is defined in terms of partially-ordered sets, and is founded on the Pareto Algebra of~\cite{pareto-algebra}.
Our framework supports compositional refinement in the sense that a component interface $C$ can safely be replaced with a component implementation $C'$ if the latter dominates the former, i.e., if $C \preceq C'$.
A running example of a video-processing system has shown how we can apply a number of basic operations (formally defined as functions) to create meaningful component compositions.
The framework enables runtime quality and resource management in embedded and cyber-physical systems.
Although we have defined several generally applicable composition patterns, typically domain or even application-specific knowledge is needed. This is certainly also true for the partially-ordered sets that form the foundation, and for the composition of the basic functions to model the composition of components.
The framework that we have presented is therefore an abstract, conceptual framework. Concrete posets, constraints, derivations and ways of composition must be created for each application domain such that the semantics of that domain are reflected properly.
We have shown how the QRML modeling language (\url{https://qrml.org}) can be used to specify such domain-specific quality and resource management models. QRML is under further development.

\section*{Acknowledgements}
Twan Basten had the pleasure of being one of the first PhD students of Jos Baeten in Eindhoven. This experience had a lasting impact on his scientific work.
Algebra is never far away. Twan and Jos have always shared an interest in describing and analyzing complex systems in mathematically rigorous and elegant ways.
In a sense, the work in this paper found its inspiration already in Twan's PhD research. Throughout his career, Jos combined research on systems engineering with managing research organizations. Both involve balancing a variety of resources and qualities to achieve desired goals. An algebra for managing qualities and resources thus seems a fitting contribution to this liber amicorum. 

\bibliographystyle{myalpha}
\bibliography{refs}

\newcommand{\etalchar}[1]{$^{#1}$}
\begin{thebibliography}{vOvdLKM00}

\bibitem[APA{\etalchar{+}}16]{AP+16}
A.~Agirre, J.~Parra, A.~Armentia, E.~Est\'evez, and M.~Marcos.
\newblock {QoS Aware Middleware Support for Dynamically Reconfigurable
  Component Based {IoT} Applications}.
\newblock {\em International Journal of Distributed Sensor Networks}, 12(4),
  2016.

\bibitem[BBS06]{BIB}
A.~Basu, M.~Bozga, and J.~Sifakis.
\newblock {Modeling Heterogeneous Real-time Components in {BIP}}.
\newblock In {\em Proc.\ 4th IEEE International Conference on Software
  Engineering and Formal Methods}, SEFM '06. IEEE, 2006.

\bibitem[BCF{\etalchar{+}}08]{B08}
A.~Benveniste, B.~Caillaud, A.~Ferrari, L.~Mangeruca, R.~Passerone, and
  C.~Sofronis.
\newblock {Multiple Viewpoint Contract-Based Specification and Design}.
\newblock In {\em Formal Methods for Components and Objects}, FMCO 2007, LNCS
  5382. Springer, 2008.

\bibitem[Box98]{com}
D.~Box.
\newblock {\em Essential {COM}}.
\newblock Addison-Wesley, 1998.

\bibitem[CEJS98]{symred}
E.~M. Clarke, E.~A. Emerson, S.~Jha, and A.~P. Sistla.
\newblock {Symmetry Reductions in Model Checking}.
\newblock In {\em Proc.\ Computer Aided Verification}, CAV 1998. Springer,
  1998.

\bibitem[dAH01a]{AH01}
L.~de~Alfaro and T.A. Henzinger.
\newblock Interface automata.
\newblock {\em SIGSOFT Softw. Eng. Notes}, 26(5), 2001.

\bibitem[dAH01b]{AH01b}
L.~de~Alfaro and T.A. Henzinger.
\newblock {Interface Theories for Component-Based Design}.
\newblock In {\em Proc.\ Embedded Software}, EMSOFT 2001, pages 148--165.
  Springer, 2001.

\bibitem[DMB08]{z3}
L.~De~Moura and N.~Bj\o{}rner.
\newblock {Z3: An Efficient SMT Solver}.
\newblock In {\em Proc.\ Tools and Algorithms for the Construction and Analysis
  of Systems}, TACAS 2008. Springer, 2008.

\bibitem[EJL{\etalchar{+}}03]{Ptolemy}
J.~{Eker}, J.~W. {Janneck}, E.~A. {Lee}, {J. Liu}, {X. Liu}, J.~{Ludvig},
  S.~{Neuendorffer}, S.~{Sachs}, and {Y. Xiong}.
\newblock {Taming Heterogeneity - the {Ptolemy} Approach}.
\newblock {\em Proceedings of the IEEE}, 91(1), 2003.

\bibitem[EL07]{pret}
S.~A. {Edwards} and E.~A. {Lee}.
\newblock {The Case for the Precision Timed ({PRET}) Machine}.
\newblock In {\em Proc.\ 44th ACM/IEEE Design Automation Conference}. ACM,
  2007.

\bibitem[GB07]{pareto-calc}
M.~{Geilen} and T.~{Basten}.
\newblock {A Calculator for {Pareto} Points}.
\newblock In {\em Proc.\ Design, Automation \& Test in Europe Conference \&
  Exhibition}, DATE '07. IEEE, 2007.

\bibitem[GBE10]{GBE10}
M.~{Garc\'\i a-Valls}, P.~{Basanta-Val}, and I.~{Est\'evez-Ayres}.
\newblock {A Component Model for Homogeneous Implementation of Reconfigurable
  Service-based Distributed Real-time Applications}.
\newblock In {\em Proc.\ 10th Annual International Conference on New
  Technologies of Distributed Systems}, NOTERE '10. IEEE, 2010.

\bibitem[GBTO07]{pareto-algebra}
M.~Geilen, T.~Basten, B.~Theelen, and R.~Otten.
\newblock {An Algebra of {Pareto} Points}.
\newblock {\em Fundamenta Informaticae}, 78(1), 2007.

\bibitem[GKN{\etalchar{+}}17]{compsoc}
K.~Goossens, M.~Koedam, A.~Nelson, S.~Sinha, S.~Goossens, Y.~Li, G.~Breaban,
  R.~van Kampenhout, R.~Tavakoli, J.~Valencia, H.~A. Balef, B.~Akesson,
  S.~Stuijk, M.~Geilen, D.~Goswami, and M.~Nabi.
\newblock {NoC-Based Multiprocessor Architecture for Mixed-Time-Criticality
  Applications}.
\newblock In S.~Ha and J.~Teich, editors, {\em Handbook of Hardware/Software
  Codesign}. Springer, 2017.

\bibitem[GTW11]{dfrefine}
M.~Geilen, S.~Tripakis, and M.~Wiggers.
\newblock {The Earlier the Better: A Theory of Timed Actor Interfaces}.
\newblock In {\em Proc.\ 14th Int.\ Conf.\ on Hybrid Systems: Computation and
  Control}, HSCC ’11. ACM, 2011.

\bibitem[HM06]{HM06}
T.~A. {Henzinger} and S.~{Matic}.
\newblock {An Interface Algebra for Real-Time Components}.
\newblock In {\em Proc.\ 12th IEEE Real-Time and Embedded Technology and
  Applications Symposium}, RTAS '06. IEEE, 2006.

\bibitem[KLSV03]{KL+03}
D.K. Kaynar, N.A. Lynch, R.~Segala, and F.~Vaandrager.
\newblock {Timed {I/O} Automata: A Mathematical Framework for Modeling and
  Analyzing Real-Time Systems}.
\newblock In {\em Proc.\ 24th IEEE International Real-Time Systems Symposium},
  RTSS '03. IEEE, 2003.

\bibitem[KSHS17]{comma}
I.~Kurtev, M.~Schuts, J.~Hooman, and D.-J. Swagerman.
\newblock {Integrating Interface Modeling and Analysis in an Industrial
  Setting}.
\newblock In {\em Proc.\ 5th International Conference on Model-Driven
  Engineering and Software Development}, MODELSWARD '17. SciTePress, 2017.

\bibitem[LLGCM10]{kneepoints}
J.~Legriel, C.~Le~Guernic, S.~Cotton, and O.~Maler.
\newblock {Approximating the Pareto Front of Multi-criteria Optimization
  Problems}.
\newblock In {\em Proc.\ Tools and Algorithms for the Construction and Analysis
  of Systems}, TACAS 2010. Springer, 2010.

\bibitem[LM87]{sdf}
E.~A. {Lee} and D.~G. {Messerschmitt}.
\newblock {Synchronous Data Flow}.
\newblock {\em Proceedings of the IEEE}, 75(9), 1987.

\bibitem[LMFS13]{LM+13}
A.~Luppold, B.~Menhorn, H.~Falk, and F.~Slomka.
\newblock {A New Concept for System-level Design of Runtime Reconfigurable
  Real-time Systems}.
\newblock {\em SIGBED Review}, 10(4), 2013.

\bibitem[LT87]{LT87}
N.A. Lynch and M.R. Tuttle.
\newblock {Hierarchical Correctness Proofs for Distributed Algorithms}.
\newblock In {\em Proc.\ 6th Annual ACM Symposium on Principles of Distributed
  Computing}, PODC '87. ACM, 1987.

\bibitem[OA18]{osgi}
OSGi Alliance.
\newblock {OSGi} release 7, 2018.

\bibitem[OMG12]{corba}
Object~Management Group.
\newblock {The Common Object Request Broker: Architecture and Specification
  Version 3.3}, 2012.

\bibitem[SKB{\etalchar{+}}15]{compsoc2}
S.~Sinha, M.~Koedam, G.~Breaban, A.~Nelson, A.~Beyranvand Nejad, M.~Geilen, and
  K.~Goossens.
\newblock {Composable and Predictable Dynamic Loading for Time-critical
  Partitioned Systems on Multiprocessor Architectures}.
\newblock {\em Microprocessors and Microsystems}, 39(8), 2015.

\bibitem[SM97]{javabeans}
Sun Microsystems.
\newblock {JavaBeans}, 1997.

\bibitem[SVDP12]{cbd17}
A.~Sangiovanni-Vincentelli, W.~Damm, and R.~Passerone.
\newblock {Taming Dr. {Frankenstein}: Contract-Based Design for Cyber-Physical
  Systems}.
\newblock {\em European Journal of Control}, 18(3), 2012.

\bibitem[vdBCH{\etalchar{+}}20]{qrml}
F.~van~den Berg, V.~Camra, M.~Hendriks, M.~Geilen, P.~Hnetynka, F.~Manteca,
  P.~S{\'{a}}nchez, T.~Bures, and T.~Basten.
\newblock {QRML: A Component Language and Toolset for Quality and Resource
  Management}.
\newblock In {\em Forum on specification \& Design Languages}, FDL 2020.
  {IEEE}, 2020.

\bibitem[vOvdLKM00]{koala}
R.~van Ommering, F.~van~der Linden, J.~Kramer, and J.~Magee.
\newblock {The {Koala} Component Model for Consumer Electronics Software}.
\newblock {\em Computer}, 33, 2000.

\bibitem[WT05]{WT05}
E.~Wandeler and L.~Thiele.
\newblock {Real-time Interfaces for Interface-based Design of Real-time Systems
  with Fixed Priority Scheduling}.
\newblock In {\em Proc.\ 5th ACM International Conference on Embedded
  Software}, EMSOFT '05. ACM, 2005.

\end{thebibliography}

\end{document}